\documentclass[pdflatex,sn-mathphys-ay]{sn-jnl}

\usepackage{graphicx}%
\usepackage{multirow}%
\usepackage{amsmath,amssymb,amsfonts}%
\usepackage{amsthm}%
\usepackage{mathrsfs}%
\usepackage[title]{appendix}%
\usepackage{xcolor}%
\usepackage{textcomp}%
\usepackage{manyfoot}%
\usepackage{booktabs}%
\usepackage{algorithm}%
\usepackage{algorithmicx}%
\usepackage{algpseudocode}%
\usepackage{listings}%
\usepackage{fullpage}
\usepackage{booktabs}
\usepackage{multirow}
\usepackage{changepage}
\usepackage{verbatim}
\usepackage{multicol}
\usepackage{physics}
\usepackage{cases}
\usepackage{enumerate}
\usepackage{mathtools}
\usepackage{color}

\newtheorem{Lemma}{Lemma}
\newtheorem{Th}{Theorem}

\newtheorem{remark}{Remark}

\numberwithin{Lemma}{section}
\numberwithin{Th}{section}
\numberwithin{Cor}{section}
\numberwithin{remark}{section}
\numberwithin{equation}{section}

\theoremstyle{definition}
\newtheorem{example}{Example}

\begin{document}

\title[Trimming to Coexistence]{Trimming to Coexistence: How Dispersal Strategies Should be Accounted for in Resource Management}

\author*[1]{\fnm{Elena} \sur{Braverman}}\email{maelena@ucalgary.ca, ORCID  0000-0002-6476-585X}

\author[1]{\fnm{Jenny} \sur{Lawson}}\email{jennifer.lawson@ucalgary.ca, ORCID  0000-0002-7881-3701}

\affil*[1]{\orgdiv{Department of Mathematics \& Statistics}, \orgname{University of Calgary}, \orgaddress{\street{2500 University Drive N.W.}, \city{Calgary}, \postcode{T2N1N4}, \state{Alberta}, \country{Canada}}}

\abstract{For two resource-sharing species we explore the interplay of harvesting and dispersal strategies, as well as their influence on competition outcomes. Although the extinction of either species can be achieved by excessive culling, choosing a harvesting strategy such that the biodiversity of the populations is preserved is much more complicated. We propose a type of heterogeneous harvesting policy, dependent on dispersal strategy, where the two managed populations become an ideal free pair, and show that this strategy guarantees the coexistence of the species. We also show that if the harvesting of one of the populations is perturbed in some way, then it is possible for the coexistence to be preserved. Further, we show that if the dispersal of two species formed an ideal free pair, then a slight change in the dispersal strategy for one of them does not affect their ability to coexist. Finally, in the model, directed movement is represented by the term $\Delta (u/P)$, where $P$ is the dispersal strategy and target distribution. We justify that once an invading species, which has an advantage in carrying capacity, chooses a dispersal strategy that mimics the resident species distribution, then successful invasion is guaranteed. However, numerical simulations show that invasion may be successful even without an advantage in carrying capacity. More work is needed to understand the conditions, in addition to targeted culling, under which the host species would be able to persist through an invasion.}




\maketitle

\section{Introduction}
\label{sec:introduction}

There are several ways in which a species' spatial distribution can be included in a mathematical description of their population dynamics. Some possibilities include describing movement through partial differential equations or specifying a patch/compartment structure using a system of either differential or difference equations. The difference equation approach is particularly useful when distinctive stages in reproduction and development can be specified, e.g. semelparous populations that only survive one reproduction event. Although the simplest models of population growth are typically based on either ordinary differential or difference equations, incorporating diffusion in these models
aims to reveal the evolutionary mechanisms responsible for population dispersal, as well as the role of species movement in competition for survival and successful invasion of new habitats.

Choosing an approach for modelling spatial distribution depends largely on the biological features and properties of the organisms. An excellent review of modelling anisotropic diffusion from probabilistic considerations, with possible position-jump and velocity-jump random walks can be found in \cite{Painter}, see also a more recent publication by \cite{Hillen_2022}.  For a discussion on how the Gaussian noise can describe gradual movements, while the L\'{e}vy flight model is more appropriate when dispersing individuals can jump long distances in short time intervals, see \cite{Lutscher} and its list of references. While in our research we only focus on the competition of mobile populations such as animals, the approach used in \cite{Painter} can describe cell movements as well. Incorporating search time in animal movement in non-homogeneous environment \citep{McKenze} leads to heterogeneous dispersal also reviewed in \cite{Painter}.

Taking into account the heterogeneity of the environment is important, as illustrated by the recent analysis of Giant Panda spread \citep{GiantPanda}. In this paper, the authors showed the importance of capturing the heterogeneity of the environment for accurate predictions. By setting the grain size and extent of the spatial area under consideration appropriately, the model is able to collect enough data on environmental heterogeneity to make an appropriate distinction between suitable and unsuitable habitats. Such an approach results in more accurate predictions.  Modern technology offers many options to assist in capturing species data at various spatial scales, many of which are non-invasive:  they do so without any disturbance to the species. Some of these methods include camera traps and other non-invasive genetic methods such as hair collection and analysis \citep{FisherCat}.

Previous mathematical study also supports taking heterogeneity of the environment into account. \cite{He} considered two diffusing populations in competition, that were identical except that one species had a constant carrying capacity $a$, while another had  a spatially heterogeneous carrying capacity $K(x)$ (here $\int_{\Omega} K(x) dx  = a$). The species with spatially heterogeneous carrying capacity outcompeted the species with constant carrying capacity. This provides mathematical support for taking spatial heterogeneity into account, since it can interact in interesting ways with diffusion.

A simple and standard way to include spatial movement in a system for species density $u_j$ at time $t$ and location $x$ is to add the term $d_j \Delta u_j$ to the right-hand side of the logistic differential equation. When the species have  a time-independent carrying capacity $K(x)$ and there is complete resource sharing, this model has the form
\begin{equation}
\label{intro_1}
\frac{\partial u_j (t,x)}{\partial t} = d_j \Delta u_j (t,x) + u_j (t,x)  \left( K(x) -    \sum_{j=1}^n u_j(t,x) \right),
\end{equation}
for $x\in\Omega$, $t>0$, $j = 1,...,n$.
However, if the only difference between the populations $u_j$, $j=1,...,n$  is the movement speed, represented by $d_j$, the species with the smallest $d_j$ wins the competition and brings all the other species to extinction \citep{Dock}. This conclusion is counterintuitive, and as a result, certain effort was taken to develop new models for dispersal in which the ability to move faster would not be detrimental for survival, including \cite{Dock}, see also recent developments of the problem by \cite{diff4}. One of the models that was developed combined random-type diffusion, which allowed for the possibility of movement from resource-rich areas to less fertile places, with an advection term aligned with the structure of the carrying capacity $K(x)$ in \eqref{intro_1}
\begin{equation}
\label{CC_equation}
\frac{\partial u_j(t,x)}{\partial t} = 
\nabla \cdot \left[ \mu_j \nabla  u_j(t,x)- \alpha u_j(t,x)\nabla K(x) \right] + u_j (t,x)  \left( K(x) -    \sum_{j=1}^n u_j(t,x) \right),
\end{equation}
for $x\in\Omega$, $ t>0$, $j = 1,...,n$, see \citep{C4,C1,C2,C3,C5}. The interplay of advection and diffusion coefficients determined a solution profile such that with low dispersal and high advection, a solution in a single-species model tends to $K(x)$. Note that, once the carrying capacity $K(x)$ becomes a solution of an equation of type \eqref{CC_equation}, it is {\em an ideal free distribution}, 
i.e. the population is spatially arranged in such a way  that any movement would decrease the total population fitness. 

A connected approach to directed diffusion, introduced by \cite{Brav}, assumed that it was the ratio $u/K$, not the population density $u$ itself, that was subject to dispersal. The advantage of this approach was that for any diffusion coefficient, the ideal free distribution becomes a solution. This approach can easily be generalized to the term $\displaystyle \Delta \left(  \frac{u}{P} \right)$, where for a system with multiple species each choosing different dispersal strategies $P_j(x)$, partially sharing resources, and co-operating in the choice of foraging areas, the model has the form
\begin{equation}
\label{intro_2}
\frac{\partial u_j (t,x)}{\partial t} = d_j \Delta  \left(  \frac{u_j (t,x) }{P_j} \right) + u_j (t,x)  \left( K(x) -    \sum_{j=1}^n u_j(t,x) \right)
\end{equation}
for $x\in\Omega$, $t>0$, $j = 1,...,n$. The advection-diffusion \eqref{CC_equation} and directed diffusion \eqref{intro_2} approaches can be unified, as was proposed by Cosner, if we combine both the dispersal and the advection terms into one spatially dependent parameter $a_j(x)$ \citep{Kamrujjaman_2016}
\begin{equation}
\label{intro_3}
\frac{\partial u_j (t,x)}{\partial t} = \nabla \cdot \left[ a_j(x) \nabla \left( \frac{u_j(t,x)}{P_j (x)} \right) \right]+ u_j (t,x)  \left( K(x) -    \sum_{j=1}^n u_j(t,x) \right).
\end{equation}
Two species competition models of the type \eqref{intro_3}, with different $P_j$ were explored by \cite{Kamrujjaman_2016}. In this paper, the authors explored two different scenarios where the only difference between the two populations was the dispersal strategy. In the first scenario, one of the species chose a dispersal strategy that was aligned with the time-independent carrying capacity $K(x)$, and as a result the other population which chose the non-optimal strategy was sent to extinction. This competitive exclusion was inevitable and did not depend on diffusion speed, unlike in \cite{Dock}. In the other scenario, the two species co-operated in their choice of grazing areas which resulted in the carrying capacity becoming a linear combination of their dispersal strategies. Here there was a coexistence equilibrium that was globally asymptotically stable \citep{Kamrujjaman_2016}, and the species formed {\em an ideal free pair} \citep{Averill} when a couple of species has an ideal free distribution.

Many interesting discoveries were made while investigating population dynamics equations with various types of dispersal. One example is the fact established by \cite{Lou_2006} that the total population size for random-type diffusion exceeds the total of the carrying capacity $K$ over the domain. While in \eqref{intro_1} and \cite{Lou_2006} the carrying capacity and the intrinsic growth rates are unified in $K$ to reduce the number of parameters, generally,
the logistic growth in a heterogeneous environment $\displaystyle r(x) u(x) \left( 1- \frac{u(x)}{K(x)} \right)$ incorporates two different functions
\citep{Kamrujjaman_2016,DeAngelis,Kor2,Kor4}. As illustrated by \cite{DeAngelis}, for such a single-species model the relation between the average population and the average carrying capacity is more complicated.

Evolutionarily stable strategies are commonly used in discussions of dispersal strategies and directed diffusion. There are two possible approaches to explaining what is meant by an evolutionarily stable strategy:
\begin{enumerate}
\item
When a home population adopts an evolutionarily stable strategy, no other population which chooses a different strategy can invade its habitat. If an invasive species chooses an evolutionary stable strategy which is different from the host species' dispersal strategy, invasion and establishment in the habitat are inevitable.
\item
If a species adopts an evolutionarily stable strategy, then small perturbations of strategy or environment cannot alter the species' survival unless an invading species chooses the optimal strategy. 
\end{enumerate}

Here we use both approaches to evaluate stability of a strategy and/or resource management and provide survival of one or both of competing species. Additionally, the ideal free distribution is naturally connected to an evolutionarily stable strategy \citep{C0}.

Let us note that above, as in the present paper, each population's dispersal strategy takes into account only two factors:
\begin{enumerate}
\item
the size of the population at each point of the domain;
\item
the environment fertility.
\end{enumerate}
In particular, the population size of the other resource-sharing species is not a factor that is taken into account. The independence of dispersal for two species was questioned by \cite{Shigesada}. Later, a relevant approach related to {\em starvation-driven diffusion} was introduced in \cite{Cho_SDD}. In this approach, the diffusion term 
for the species $u$ competing with $v$, has the form $\Delta [ \gamma(s(u,v,K)) u ]$. 
Here $s(\cdot,\cdot,K)$ is a satisfaction function determined, for example, by the ratio of the carrying capacity $K$ to the sum (or a linear combination in the case of partial resource sharing) of population densities $u$ and $v$ at the point $x$. $\gamma(\cdot)$ is usually a step function which activates movement for low values of satisfaction only. This idea was explored and extended to predator-prey systems, see \citep{diff3,diff1} and references therein. Some approaches emphasize the high mobility of predators considering prey-dependent dispersal \citep{diff2,prey_dep_dispersal}. For some other models describing knowledge-based animal movement, and an overview of some open problems, see \citep{diff5}.

Within the field of ecology, understanding the relationship between a population and its environment is crucial for being able to effectively manage the population. Part of population management is tracking and setting appropriate levels of harvesting corresponding to deliberate human-caused mortality. Research shows that harvesting effort and intensity are not constant across large geographical areas \citep{CabanellasReboredoetal2017} and can depend on other environmental, social, and habitat factors \citep{Karnsetal2016}.   Recreational fisheries, for example, are strongly location-dependent, with fishers favouring locations with appropriate weather conditions that are closer to shore \citep{CabanellasReboredoetal2017}. Further, in \citep{Karnsetal2016},  the authors explored the mortality and harvest density of white-tailed deer in Ohio in 2007-2012. The proportion of farmland, other environmental conditions, and  the number of deer hunting permits appeared to be the crucial factors in influencing harvest density. As a result, the authors proposed that, instead of managing based on administrative division (counties) that may have different characteristics, management should be broken up into 6 different more environmentally homogeneous zones  \citep{Karnsetal2016}.

Mathematical models that do not consider spatial heterogeneity run the risk of over-estimating harvesting limits such as the maximum sustainable yield~(MSY)  \citep{TakashinaMougi2015}, or not recognizing when a population is already over-exploited \citep{Sirenetal2004}. In \cite{Sirenetal2004}, the authors collected data on big game hunting conducted by an indigenous community in Ecuador. In this community, hunting effort was 37 times higher near the village than in the furthest region from the village. Using a spatial discretization of 7 different zones, the authors noted that several species showed signs of being over-harvested in their respective zones while the averaged situation that was previously used did not raise any signs of alarm. We use a continuous description of the environment, and not a patch structure like the ones found in  \citep{TakashinaMougi2015}, to alleviate the cost of not considering spatial heterogeneity.

Even the effect to which harvesting disturbs an environment, and its resident non-target species, depends largely on location. In \cite{Gutierrezetal2004} the authors examined the effects of harvesting stout razor clams on non-target species living on the ocean floor. They hypothesized that the extent to which a community of non-targeted species would be effected, is a function of the density of another species, burrowing crabs. These burrowing crabs often burrow in areas that have been disturbed by harvesting and can interfere with the original non-target resident community. However, what they found was that the density of burrowing crabs had no effect on the extent of the disruption that the existing community experienced. Instead, the response of the non-target resident community was mostly spatially dependent \citep{Gutierrezetal2004}, presumably stipulated by the heterogeneity of the environment.

Harvesting can also be used to control and eliminate invasive species. \cite{FresardRoparsCollet2014} used optimal control to investigate how a sustainable harvest could be achieved, when a species with market value is invaded by another species with no market value. Using harvesting of both species, the authors showed that a sustainable harvest was possible, but only if it was implemented early. That being said, when the invasive species is already established in a large area, its elimination via harvesting may be prohibitively expensive, or impossible. Incentivized harvest programs are sometimes implemented to harness the efforts of recreational and commercial fishers in assisting with removal efforts, although these programs can also have unintended effects \citep{PaskoGoldberg2014}. It is possible for incentivized harvest programs to be successful, but for that to happen several steps should be taken including having a defined management objective and plan, preventing re-introduction, and monitoring of unintended outcomes \citep{PaskoGoldberg2014}.

Other management strategies involve the targeted removal or culling of invasive species in a specific area as in \cite{BakerBode2013} where poison baits installed outside of the perimeter of the conservation area reduced the invasion levels. The optimal spatial distribution of poison baits around an area of high conservation was computed. This would allow for targeted baiting of invasive species around valuable landscapes such as national parks, threatened ecosystems, or other habitats of vulnerable species, to reduce the harmful incursion of these species. It is clear that harvesting has a natural spatial component, and that management via harvesting requires careful thought as to its spatial arrangement. More research is required to investigate how to optimally arrange harvesting such that sustainability of desired populations is maintained.

In this paper, we explore the possibility that both populations can be harvested, and also investigate both the impact of harvesting and the impact of choice of diffusion strategy on the competition outcomes. For a competing system with directed diffusion, we set up the model and get auxiliary results in Section~\ref{sec:auxil}. This will allow us to tackle the following major problems which are addressed in Section~\ref{sec:main}:
\begin{enumerate}
\item
Can we adapt harvesting strategies for one or both species not only to the carrying capacity, but also to the dispersal strategies of the competing species in such a way that both populations coexist? (Section \ref{subsec:trimming})
\item
If the dispersal strategy of an invading population tries to mimic the distribution of the resident species, does this promote successful invasion? (Section \ref{subsec:mimicking}) 
\item
If two species coexist in an ideal free pair, will small changes to one of their dispersal strategies still allow them to coexist or simply to survive? (Section \ref{subsec:deviating})
\item
If certain levels and spatial profiles of harvesting guarantee coexistence, will this be changed by small deviations from the harvesting policy? (Section \ref{subsec:harvesting})
\end{enumerate} 
The four subsections of   Section~\ref{sec:main} answer these questions. Section~\ref{sec:numerical} illustrates the results in Section~\ref{sec:main} with numerical simulations. In Section~\ref{sec:discussion} we discuss both the results of the current paper and relevant open problems.
Most of the proofs are postponed to Appendix~A, except the justification of Theorem~\ref{coexist_dev}, which is more involved and is presented in Appendix~B.

\section{Preliminaries and Auxiliary Results}
\label{sec:auxil}

Consider the model with a potential of harvesting in both equations \vspace{-6mm}
\begin{flushleft}
\begin{equation}
\label{system}
\begin{cases}
 \displaystyle \frac{\displaystyle \partial u}{\displaystyle \partial t}
=  
 \displaystyle  \nabla \cdot \bigg[ a_1(x)  \nabla \bigg( \frac{\displaystyle u(t,x)}{\displaystyle P(x)}\bigg) \bigg]
\displaystyle + u(t,x) \left(K_1(x)-  u(t,x) - v(t,x) \right) - E_1(x) u(t,x), \vspace{1mm}
\\
\displaystyle \frac{\displaystyle \partial v}{\displaystyle \partial t}
= 
\displaystyle  \nabla \cdot \left[ a_2(x)  \nabla \left( \frac{\displaystyle v(t,x)}{\displaystyle Q(x)}\right) \right]
 \displaystyle + v(t,x) \left(K_2(x) - u(t,x) - v(t,x) \right) - E_2(x) v(t,x),\\
t>0,\;x\in{\Omega},
\\
\displaystyle  \frac{\displaystyle \partial }{\displaystyle \partial n} \bigg( \frac{u}{P}\bigg) = \displaystyle  \frac{\displaystyle \partial }{\displaystyle \partial n} \bigg( \frac{v}{Q}\bigg) = 0, ~x\in\partial\Omega. 
\end{cases}
\end{equation}
\end{flushleft}
\noindent Denoting in \eqref{system}
\begin{equation}
\label{no_harvest}
m_i(x) := K_i(x) - E_i(x), \quad i=1,2,
\end{equation}
we get the system
\begin{align}
\label{ssystem}
\begin{cases}
\displaystyle \frac{\displaystyle \partial u}{\displaystyle \partial t}
=   \displaystyle  \nabla \cdot \left[ a_1(x)  \nabla \left( \frac{\displaystyle u(t,x)}{\displaystyle P(x)}\right) \right]
 \displaystyle + u(t,x) \left(m_1(x)-  u(t,x) - v(t,x) \right), \vspace{1mm}
\\
\displaystyle \frac{\displaystyle \partial v}{\displaystyle \partial t}
=   \nabla \cdot \left[ a_2(x)  \nabla \left( \frac{\displaystyle v(t,x)}{\displaystyle Q(x)}\right) \right]  \displaystyle + v(t,x) \left(m_2(x) - u(t,x) - v(t,x) \right),
\\
t>0,\;x\in{\Omega},
\\
\displaystyle  \frac{\displaystyle \partial }{\displaystyle \partial n} \bigg( \frac{u}{P}\bigg) = \displaystyle  \frac{\displaystyle \partial }{\displaystyle \partial n} \bigg( \frac{v}{Q}\bigg) = 0, ~x\in\partial\Omega,
\end{cases}
\end{align}  
where  $m_i(x)$ the ``after-harvesting carrying capacities" are defined in \eqref{no_harvest}.
We assume that the domain $\Omega$ is an open bounded region in $\mathbb{R}^n$ with ${\partial}{\Omega}\in C^{2+\alpha}$, $\alpha>0$. The functions $m_i(x), a_i(x)$, $i=1,2$, $P(x),Q(x)$ are continuous and positive on $\overline{\Omega}$. $P,Q \in C^{2 + \alpha} (\overline{\Omega})$, $a_i \in C^{1 + \alpha}(\overline{\Omega})$, and $m_i \in C^{\alpha} (\overline{\Omega})$.

Let the functions $u^{\ast}$ and $v^{\ast}$ be solutions of the single-species stationary models
corresponding  to the first and the second equations in  \eqref{ssystem}, we get
\begin{equation}\label{eq_u}
\begin{cases}
\nabla \cdot \left[ a_1(x)  \nabla \left( \frac{\displaystyle u^{\ast}(x)}{\displaystyle P(x)}\right) \right]
+  u^{\ast}(x) \left(  m_1(x) -   \displaystyle u^{\ast}(x)  \right) = 0,\; x\in\Omega,\;\\
\displaystyle \frac{\displaystyle  \partial}{\displaystyle \partial n} \left(\frac{u^{\ast}}{P} \right)  =0,\; x\in\partial\Omega 
\end{cases}
\end{equation}
and
\begin{equation}\label{eq_v}
\begin{cases}
\nabla \cdot \left[ a_2(x) \nabla \left( \frac{\displaystyle v^{\ast}(x)}{\displaystyle Q(x)}\right) \right]
+   v^{\ast}(x)\left(m_2(x)  -  \displaystyle v^{\ast}(x)  \right) = 0,\; x\in\Omega,\;\\
\displaystyle \frac{\displaystyle  \partial}{\displaystyle \partial n}   \left(  \frac{v^{\ast}}{Q} \right)=0,\; x\in\partial\Omega ,
\end{cases}
\end{equation}
respectively.

Note that some of the inequalities we derive in this section have already been obtained in this or some equivalent form previously, but we present the statements with justification for completeness of the presentation.
In future, we will need the following two statements  that in some sense complement each other, and their modifications. 

\begin{Lemma}\label{Lsteady2}
Let $u^{\ast}$ be a positive solution of \eqref{eq_u}, then
\begin{equation}\label{eq_u2_1}
\int \limits_\Omega   P(x)\left(   u^*(x)  - m_1 (x)  \right)\,dx=
\int \limits_\Omega \frac{a_1(x)|\nabla (u^*/P)|^2}{(u^*/P)^2}\,dx \geq 0.
 \end{equation}
If $P(x)$ and $m_1(x)$ are linearly independent on $\Omega$   then
\begin{equation}\label{eq_u2}
\int \limits_\Omega P(x)\left(  u^*(x) - m_1(x)  \right)\,dx > 0. 
\end{equation}
\end{Lemma}

\begin{remark}
In the particular case of the regular diffusion $P \equiv c$ which is constant over $\Omega$, inequality \eqref{eq_u2} means that in heterogeneous environment, 
the average population is higher than the carrying capacity. This aligns with the results shown by \cite{Lou_2006}.
\end{remark}

Analyzing \eqref{eq_v}, we get a similar result 
for $v^{\ast}$ whenever $Q$ and $m_2$ are linearly independent, i.e.
\begin{equation*}
\int \limits_\Omega Q(x)\left(  v^*(x) - m_2(x)  \right)\,dx > 0. 
\end{equation*}

\begin{Lemma}
\label{Lsteady1}
Suppose that $u^{\ast}$ is a positive solution of \eqref{eq_u}, while
$P(x)$ and $m_1(x)$ are such that $\nabla \cdot [a_1(x)\nabla (m_1(x)/P(x))] \not\equiv 0$ on $\Omega$. 
Then
\begin{equation}\label{eq_est1abc}
\int\limits_{\Omega}m_1 \left( m_1(x) - u^{\ast}(x) \right)  \,dx   > 0.
\end{equation}
\end{Lemma}

Similarly, if 
$\nabla \cdot [a_2(x)\nabla (m_2(x)/Q(x))] \not\equiv 0$ on $\Omega$, we have 
 \begin{equation*}
 \int \limits_\Omega m_2(x)\left(m_2(x) -   v^{\ast}(x)  \right)\,dx>0.
 \end{equation*}

The following results deal with the case when there exists a positive stationary coexistence equilibrium of \eqref{ssystem} which satisfies
\begin{equation}
\label{coexist_system}
\begin{cases}
 \displaystyle  \nabla \cdot \left[ a_1(x) \nabla \left( \frac{u_s(x)}{P(x)} \right)  \right]
+ u_s(x) \left(m_1(x)- u_s(x) - v_s(x) \right) = 0, \vspace{1mm}
\\
 \displaystyle  \nabla \cdot \left[  a_2(x) \nabla \left( \frac{v_s(x)}{Q(x)} \right)  \right]
+ v_s(x) \left(m_2(x)- u_s(x) - v_s(x) \right) = 0, & 
 \vspace{2mm} \\
  t>0,\;x\in{\Omega},\\
\displaystyle \pdv{}{n}\bigg(\frac{u_s}{P} \bigg) = \pdv{}{n}\bigg( \frac{v_s}{Q}\bigg) = 0, x\in\partial \Omega.
\end{cases}
\end{equation}

\begin{Lemma}\label{Lsteady2abc}
Let $(u_s,v_s)$ be a positive stationary coexistence solution of \eqref{ssystem} satisfying \eqref{coexist_system}.
If  $\displaystyle \nabla \cdot [a_1 \nabla  (u_s /P)] \not\equiv 0$  on $\Omega$ then
\begin{equation*}
\int \limits_\Omega   P(x)\left(   u_s(x) + v_s (x)   - m_1 (x)  \right)\,dx >0 .
\end{equation*}
If  $\displaystyle \nabla \cdot [a_2 \nabla ( v_s /Q)] \not\equiv 0$  on $\Omega$ then
\begin{equation*}\label{eq_u2_1_coexistence}
\int \limits_\Omega   Q(x)\left(   u_s(x) + v_s (x)   - m_2  (x)  \right)\,dx >0 .
 \end{equation*}
\end{Lemma}

\begin{Lemma}
\label{Lsteady1abc}
Let $(u_s,v_s)$ be a positive stationary coexistence solution of \eqref{ssystem} satisfying \eqref{coexist_system}. If $m_1(x) \geq m_2(x)$  on $\Omega$  and $u_s+v_s \not\equiv m_1$, then
\begin{equation}\label{eq_est1}
\int\limits_{\Omega}m_1(x)  \left( m_1 (x) -  u_s(x) -v_s(x) \right) \,dx> 0.
\end{equation}
If $m_2(x) \geq m_1(x)$ on $\Omega$ and $u_s+v_s \not\equiv m_2$ then
\begin{equation*}
\int\limits_{\Omega}m_2 (x)  \left( m_2(x) -  u_s(x) -v_s(x) \right) \,dx> 0.
\end{equation*}
\end{Lemma}

These auxiliary results will allow us to analyze stationary solutions (zero, semi-trivial and coexistence equilibria) of \eqref{ssystem}.
Problem \eqref{ssystem} is a monotone dynamical system \citep{CC,Hsu,Pao,Hs}. If all the equilibrium solutions but one are unstable, the remaining one should be globally asymptotically stable.
The trivial  equilibrium for positive $m_1,m_2$ is unstable \citep{CC,Kor2,Kor4}.

\begin{Lemma}
\label{zero_equilibrium}
The zero solution $(0,0)$ of \eqref{ssystem} is an unstable repelling equilibrium.
\end{Lemma}

Without loss of generality, adapting the coefficients $a_1$ accordingly, we can assume that the dispersal strategies are normalized in the sense that
\begin{equation}
\label{normal}
\int_{\Omega} P(x) ~dx =1, \quad   \int_{\Omega} Q(x) ~dx =1.
\end{equation}

\begin{Lemma}\label{semi_p1}
Suppose that \eqref{normal} holds, $P$ and $m_1$  are  linearly independent on $\Omega$ functions, $m_2(x) \geq m_1(x)$ on $\Omega$
and $m_1(x)\equiv \alpha P+\beta Q$ for some $\alpha \geq 0$, $\beta>0$.
Then the semi-trivial steady state $(u^{\ast}(x),0)$ of  \eqref{ssystem} is unstable.
\end{Lemma}

\begin{remark}
Note  that in the case when $(\alpha P, \beta Q)$ is a stationary solution of \eqref{ssystem}, it is called {\bf an ideal free pair} \citep{Averill}.
\end{remark}

As the situation is symmetric, we get from Lemma~\ref{semi_p1} instability of the second semi-trivial equilibrium.

\begin{Lemma}\label{semi_p2}
Suppose that $Q$ and $m_2$  are  linearly independent on $\Omega$ functions, $m_2(x) \leq m_1(x)$ on $\Omega$
and $m_2(x)\equiv \alpha P+\beta Q$ for some $\alpha>0$, $\beta>0$.
Then the semi-trivial steady state $(0,v^{\ast}(x))$ of \eqref{ssystem}  is unstable.
\end{Lemma}

For completeness, let us also consider the case when the carrying capacity is proportional to one of the dispersal strategies only (and non-proportional to the other one).

\begin{Lemma}\label{semi_proportional}
Suppose that   $m_1  \equiv  m_2 \equiv m \equiv \alpha P$, where $\alpha>0$, $Q$ and $P$  are  linearly independent on $\Omega$ functions.
Then the semi-trivial steady state $(0,v^{\ast}(x))$  of \eqref{ssystem}  is unstable, and there is no coexistence equilibrium.
\end{Lemma}

The next seminal result by \cite{Hsu} classifies all possible equilibria for a monotone dynamical system, \eqref{ssystem} is a particular case of such a system.

\begin{Lemma}\label{equil_charac}
Let the zero equilibrium be a repeller of \eqref{ssystem}. Then, for system \eqref{ssystem},
if all but one equilibrium solutions are unstable, the remaining equilibrium is a global attractor for all solutions of \eqref{ssystem} with non-negative non-trivial, in both $u$ and $v$ initial conditions.
\end{Lemma}

\section{Main Results}
\label{sec:main}

\subsection{Competition outcomes without harvesting - dispersal  influence}
\label{subsec:intro}

\begin{Th}
\label{stability_coexistence_1}
Suppose that  for some $\alpha>0$, $\beta>0$, $m_2 \equiv m_1 \equiv m \equiv \alpha P+\beta Q > 0$ on $\Omega$, where  $(P,Q)$, $(P,m)$ and $(Q,m)$ are pairs of functions that are linearly independent on $\Omega$. 
Then, all solutions of \eqref{ssystem} with non-negative non-trivial, in both arguments, initial conditions converge to the coexistence equilibrium $(\alpha P,\beta Q)$.
\end{Th}

\begin{Th}\label{stability_exclusion_1}
Suppose that 
$(P,Q)$ is a  linearly independent on $\Omega$ pair of functions, and $m_2(x) \equiv m_1(x) \equiv m(x) = \alpha P$ on $\Omega$
for some $\alpha>0$. 
Then all solutions of \eqref{ssystem} with non-negative non-trivial initial conditions converge to the semi-trivial equilibrium $(\alpha P,0)$.
\end{Th}

\subsection{Trimming to coexistence - harvesting or resource management}
\label{subsec:trimming}

Next, let us return to original harvested system \eqref{system} and explore how to design a harvesting strategy that preserves biodiversity by ensuring that both species survive. In \cite{Ilmer}, the authors looked at the influence of harvesting on competition outcomes for two species, each of which adopted carrying capacity driven diffusion. The authors derived some estimates of the necessary relationship between the species' harvesting rates, for the purpose of either promoting coexistence, or promoting competitive exclusion. However, the setting was different to the one described by \eqref{system}. First of all, in \eqref{system}, the carrying capacity and the spatially dependent intrinsic growth rate are assumed to be equal. This choice was made in this paper to decrease the number of model parameters, but it was not made in \cite{Ilmer}. One choice that was made in \cite{Ilmer} however, that was not made in this paper, was to make the harvesting rate proportional to the growth rate, once again to decrease the number of model parameters. This resulted in the harvesting rates having similar spatial profile, with the only difference being the magnitude. In the following section, we aim to give more general results on the range of possible $E_1, E_2$ (not assumed to be proportional) that can be chosen such that both populations coexist.

\begin{Th}\label{stability_coexistence_2}
Suppose that $P,Q,K_1,K_2  > 0$, where $P$ and $Q$ satisfy \eqref{normal} and are not identical. Then there exist $E_1$ and $E_2$ such that all solutions of \eqref{system} with non-negative non-trivial initial conditions converge to a coexistence equilibrium $(u_s,v_s)$.
Moreover, for any prescribed $\mu > 0$, we can find $E_1(x;\mu)$, $E_2(x;\mu)$ such that 
\begin{equation}
\label{ratio}
\int_{\Omega} u_s \, dx = \mu \int_{\Omega} v_s \, dx.
\end{equation}
\end{Th}

\begin{remark}
The result of Theorem~\ref{stability_coexistence_2} follows from  Theorem~\ref{stability_coexistence_1}, once we can  choose $\alpha>0$,
$\beta>0$ such that $\alpha/\beta = \mu$ and the harvesting efforts $E_1=K_1 - \left( \alpha P + \beta Q \right)$, $E_2= K_2 -  (\alpha P + \beta  Q )$
are non-negative functions on $\Omega$.
In the proof  of Theorem~\ref{stability_coexistence_2} , we choose harvesting such that it requires the smallest amount of intervention possible. Often in discussions of harvesting, the question of the maximum amount of harvesting allowed arises. To answer this question we would first need to set bound such that below that bound, the species is effectively extinct. This bound should be derived from ecological and species data.
\end{remark}

\subsection{Does mimicking a resident's dispersal assist invading species?}
\label{subsec:mimicking}

If a resident population exists in relative isolation for a substantial time, because of the global attractivity of the unique positive stationary solution, its eventual dispersal is aligned with the solution $u^{\ast}$ of \eqref{eq_u}. Can an invading species benefit from mimicking the resident's spatial distribution, while also having some advantage in the carrying capacity? The following statement gives an affirmative answer to this question.

\begin{Th}
\label{mimic_invasion}
If an invasive species $v$ chooses its dispersal strategy to be aligned with the stationary solution for the host species $u$, i.e. $Q = u^{\ast}$, where $u^{\ast}$ is the solution of \eqref{eq_u}, and its carrying capacity $m_2 \geq m_1$ on $\Omega$, while exceeding it on some subdomain  $\Omega_1 \subseteq \Omega$, the successful invasion is guaranteed.
\end{Th}

\begin{remark}
When $m_1>m_2$, Example \ref{ex:mimick} will show that we can still observe successive invasion for certain choices of carrying capacity where the resources are highly concentrated in one area.
\end{remark}

\subsection{How do slight changes in a species' dispersal strategy influence competition outcomes?}
\label{subsec:deviating}

Both harvesting and dispersal strategies can be subject to some small perturbations. To highlight the specific impact of dispersal strategy on competition outcomes, we assume that the carrying capacities of the two species coincide, and consider a slight modification of \eqref{ssystem}, where the second species chooses dispersal strategy $Q_1$, which is slightly different from $Q$:
\begin{equation}
\label{system2}
\begin{cases}
\displaystyle \frac{\displaystyle \partial u}{\displaystyle \partial t}
=   \displaystyle  \nabla \cdot \left[ a_1(x)  \nabla \left( \frac{\displaystyle u(t,x)}{\displaystyle P(x)}\right) \right]
 \displaystyle + u(t,x) \left(m(x)-  u(t,x) - v(t,x) \right), \vspace{1mm}
\\
\displaystyle \frac{\displaystyle \partial v}{\displaystyle \partial t}
=   \nabla \cdot \left[ a_2(x)  \nabla \left( \frac{\displaystyle v(t,x)}{\displaystyle Q_1(x)}\right) \right]  \displaystyle + v(t,x) \left(m(x) - u(t,x) - v(t,x) \right),
\\
t>0,\;x\in{\Omega},
\vspace{2mm}
\\
\displaystyle \pdv{}{n}\bigg(\frac{u}{P} \bigg) = \pdv{}{n}\bigg(\frac{v}{Q_1} \bigg) = 0,~x\in
{\partial}{\Omega}.
\end{cases}
\end{equation}  
Here $Q_1>0$ satisfies the same smoothness conditions as $Q$, and is normalized, so as to focus on strategy and not magnitude
\begin{equation}
\label{normal_add}
 \int_{\Omega} Q_1(x) ~dx =1.
\end{equation}

\begin{Lemma}\label{semi_perturbation}
Suppose that \eqref{normal} and \eqref{normal_add} hold, $m_1\equiv m_2 \equiv m$, $\displaystyle \nabla \cdot \left[ a_1(x)  \nabla \left( \frac{\displaystyle m(x)}{\displaystyle P(x)}\right) \right] \not\equiv 0$ on $\Omega$, 
and $m(x)\equiv \alpha P+\beta Q$ for some $\alpha > 0$, $\beta>0$. 
Then there exists $\varepsilon>0$ such that $\displaystyle \int_{\Omega} \left| Q(x)-Q_1(x) \right| ~dx < \varepsilon$ guarantees that survival of the second species
in \eqref{system2}  is preserved under this strategy perturbation.
\end{Lemma}

Now we can discuss two possible scenarios when deviation from a cooperative dispersal strategy can guarantee coexistence or provide competitive exclusion of the other species with the unchanged strategy.

\begin{Th}\label{stability_coexistence_3}
Suppose that \eqref{normal} and \eqref{normal_add} hold, 
$P$ and $m$  are  linearly independent on $\Omega$ functions, 
and $m(x)\equiv \alpha P+\beta Q$ for some $\alpha > 0$, $\beta>0$. Then  
\begin{description}
\item[(1)] 
There exists a deviation bound  $\varepsilon_d >0$ such that for all $Q_1>0$, once the inequality
\begin{equation}
\label{perturbation_strategy}
 \int \limits_\Omega    |Q(x)-Q_1(x)| < \varepsilon_d
\end{equation}
holds, the second species guarantees that its survival sustains in \eqref{system2} with the deviating dispersal strategy. \\
\item[(2)] 
There is a specific deviation of the strategy (aligned with the ideal free distribution) such that the second species can guarantee competitive exclusion of the first species.
\end{description}
\end{Th}

\begin{remark}
\label{rem_exclusion}
Both estimates in Theorem~\ref{stability_coexistence_3} are sufficient. We'll present numerical simulation to illustrate how close the strategy should be to the carrying capacity profile to guarantee competitive exclusion.
There is a very specific winning strategy for the second species. As the deviation from this optimal dispersal grows, other scenarios can be observed, up to competitive exclusion of the second species.
\end{remark}

Theorem~\ref{stability_coexistence_3} only states that, once a non-perturbed system has an ideal free pair as a solution, a slight change of the dispersal strategy of the second species cannot bring it to extinction. The following statement considers persistence of the first species, if the deviation in the dispersal strategy does not lead to significant enough change of the solution to \eqref{eq_v}.

\begin{Lemma}
\label{lemma_new_deviation}
Suppose that \eqref{normal} and \eqref{normal_add} hold, $m_1\equiv m_2 \equiv m$, $\displaystyle \nabla \cdot \left[ a_2(x)  \nabla \left( \frac{\displaystyle m(x)}{\displaystyle Q(x)}\right) \right] \not\equiv 0$ on $\Omega$, 
and $m(x)\equiv \alpha P+\beta Q$ for some $\alpha > 0$, $\beta>0$. Let $v^*$ and $v_1^*$ be solutions of problem \eqref{eq_v} and \eqref{eq_v} where $Q$ is substituted with $Q_1$, respectively.
Then there exists $\varepsilon>0$ such that $\displaystyle  \left| v^*(x)-v_1^*(x) \right| < \varepsilon$ for any $x \in \Omega$ guarantees  survival of the first species
in \eqref{system2}.
\end{Lemma}

So far, we have talked about when we can expect either the $u$ or the $v$ population to survive under small deviations to the dispersal strategy. 
We can further see that if the original dispersal strategy makes the populations coexist as an ideal free pair, then small deviations to that dispersal strategy will not affect the populations ability to coexist. 
To show this, we impose a form of the deviated dispersal strategy
\begin{equation}
\label{Q1_form}
Q_1 = Q + d \cdot g(x), \quad\text{where } \int_{\Omega} g(x) dx = 0, d\in\mathbb{R}.
\end{equation}
Here $g(x)$ could be interpreted as the spatial distribution of the deviation, whereas $d$ could be interpreted as the magnitude. 
In the following theorem, we assume low variability of the ratio of the carrying capacity to the dispersal strategy of the second species $m/Q$. 
The conditions  $\displaystyle \nabla \cdot \left[ a_2(x)  \nabla \left( \frac{\displaystyle m(x)}{\displaystyle Q(x)}\right) \right] \not\equiv 0$  and $\displaystyle \frac{\displaystyle \sup_{x\in\Omega} \frac{m}{Q}}{\displaystyle \inf_{x\in\Omega} \frac{m}{Q}} <2$ below, indicate that while the strategy $Q$ is not aligned with $m$, it must not deviate too much from its profile, so that the variation of the ratio is less than twice over $\Omega$.
The proof of Theorem~\ref{coexist_dev} is postponed to Appendix~B.

\begin{Th}
\label{coexist_dev}
Suppose that \eqref{normal} and \eqref{normal_add} hold, $m_1\equiv m_2 \equiv m$, $\displaystyle \nabla \cdot \left[ a_2(x)  \nabla \left( \frac{\displaystyle m(x)}{\displaystyle Q(x)}\right) \right] \not\equiv 0$ on $\Omega$, 
and $m(x)\equiv \alpha P+\beta Q$ for some $\alpha,\beta>0$. Additionally assume that
\begin{equation*}
    \displaystyle \frac{\displaystyle \sup_{x\in\Omega} \frac{m}{Q}}{\displaystyle \inf_{x\in\Omega} \frac{m}{Q}} <2.
\end{equation*}
Let \eqref{Q1_form} hold. Then there exists a $d$, where $0<|d|<<1$, such that the coexistence of solutions to \eqref{system2} is preserved.
\end{Th}

For completeness, let us also consider the case when the species with a deviating dispersal strategy originally had a decisive advantage over the competitor.

\begin{Th}\label{semi_perturbation_aligned}
Suppose  that \eqref{normal} and \eqref{normal_add} hold, $m_1\equiv m_2 \equiv m$, $\displaystyle \nabla \cdot \left[ a_1(x)  \nabla \left( \frac{\displaystyle m(x)}{\displaystyle P(x)}\right) \right] \not\equiv 0$ on $\Omega$, while
$m(x)\equiv \beta Q$ for some $\beta>0$. 
Then there exists $\varepsilon>0$ such that $\displaystyle \int_{\Omega} \left| Q(x)-Q_1(x) \right| ~dx < \varepsilon$ guarantees survival of the second species
in \eqref{system2}  under this strategy perturbation.
\end{Th}


If a population has chosen its dispersal strategy such that it competitively excludes the other population, then it is not guaranteed that small changes will have no effect. The following remark demonstrates that if the dispersal strategy changes by a little bit, it is possible that that small change may allow the other population to survive and coexist.

\begin{remark}
\label{ex:counterexample}
Assume that the assumptions of  Theorem~\ref{semi_perturbation_aligned} are satisfied.
Let $\varepsilon \in (0,1)$ be small enough, such that $Q-\varepsilon P >0$ everywhere on $\Omega$ and consider
$$
Q_1 = \frac{1}{1-\varepsilon} \left( Q-\varepsilon P \right).
$$
Note that  $$\displaystyle  |Q-Q_1|   = \left| \left(1-  \frac{1}{1-\varepsilon}  \right) Q + \frac{\varepsilon} {1-\varepsilon} P \right|
\leq \frac{\varepsilon} {1-\varepsilon} (P+Q),
 $$
therefore
 $\displaystyle \int_{\Omega} \left| Q(x)-Q_1(x) \right| ~dx =  \frac{2\varepsilon} {1-\varepsilon} \to 0$ as $\varepsilon \to 0$.
For such $Q_1$, we have 
$$
\beta (1-\varepsilon)  Q_1(x)  +  \beta  \varepsilon P(x) = \beta Q(x) = m(x).
$$
Thus $( \beta  \varepsilon P,  \beta (1-\varepsilon)  Q_1)$ is a globally attractive free pair, see Theorem~\ref{stability_coexistence_1}.
This guarantees coexistence and not a competitive exclusion of the first species.
\end{remark}

\subsection{Are  slight changes  in harvesting effort and distribution crucial for coexistence?}
\label{subsec:harvesting}

We have shown that it is possible to harvest in such a way that two populations can be pushed to coexistence. But important to the context of sustainable management, is whether these coexistence states are stable under perturbations to the harvesting intensity. Thus, this topic has received a high level of scrutiny. For the current model, it was explored in \cite{Ilmer} for $E_i= \alpha_i m(x)$ in the following context: one of the species is harvested at a fixed rate $\alpha_1=\alpha_0$, while for the other species $\alpha_2$ is slightly increased, with $\alpha_2 \in (\alpha_0, \alpha_0 + \alpha^{\ast})$. In contrast, here we don't assume any limitations on the spatial profile of the harvesting perturbation. This means we include considerations of management redistribution. For example, for the same number of issued licenses without strict limitations on the fishery location this can correspond to over-exploitation of certain areas with others being under-used. Generally, harvesting redistribution can influence competition outcomes.

Consider a perturbation of \eqref{system}, where only the space-dependent  harvesting rate of the second species is perturbed by $E_p$
\begin{equation}
\label{system3}
\begin{cases}
\begin{aligned}
\displaystyle \frac{\displaystyle \partial u}{\displaystyle \partial t}
= &   \displaystyle  \displaystyle  \nabla \cdot \bigg[ a_1(x)    \nabla  \bigg( \frac{\displaystyle u(x)}{\displaystyle P(x)}\bigg) \bigg] \displaystyle \vspace{1mm}
\\ &+ u(t,x) \left(K_1(x)-  u(t,x) - v(t,x) \right) - E_1(x) u(t,x), 
\end{aligned}
\vspace{2mm}
\\
\begin{aligned}
\displaystyle \frac{\displaystyle \partial v}{\displaystyle \partial t}
=  & \displaystyle  \displaystyle  \nabla \cdot \bigg[ a_2(x)  \nabla  \bigg( \frac{\displaystyle v(x)}{\displaystyle Q(x)}\bigg) \bigg] \displaystyle \vspace{1mm}
\\&+ v(t,x) (K_2(x) - u(t,x)  - v(t,x)   - E_2(x) v(t,x) - E_p(x)v(t,x),
\end{aligned}
\vspace{2mm}
\\ 
t>0,\;x\in{\Omega},
\\ \displaystyle
\pdv{}{n}\bigg( \frac{u}{P}\bigg) = \pdv{}{n}\bigg(\frac{v}{Q} \bigg) =0,~x\in
{\partial}{\Omega}.
\end{cases}
\end{equation}  
Also, not to get involved into multi-factor evaluation to determine competition outcome, we assume that, to achieve coexistence, as in Section~\ref{subsec:trimming},
\begin{enumerate}
\item
after harvesting, the same carrying capacity for the two species is provided $m_1(x) = K_1(x) - E_1(x) \equiv m_2 =  K_2(x) - E_2(x) \equiv  m(x)$;
\item
to promote coexistence with the original harvesting efforts $E_1,E_2$, $u$ and $v$ form an ideal free pair corresponding to $m(x) = \alpha P(x) + \beta Q(x)$, where $\alpha,\beta$ are positive constants.
\end{enumerate}

Under these conditions, we get 
\begin{equation}
\label{ssystem3}
\begin{cases}
\displaystyle \frac{\displaystyle \partial u}{\displaystyle \partial t}
=   \displaystyle  \nabla \cdot \left[ a_1(x)  \nabla \left( \frac{\displaystyle u(t,x)}{\displaystyle P(x)}\right) \right]
 \displaystyle + u(t,x) \left(m(x)-  u(t,x) - v(t,x) \right), \vspace{1mm}
\\
\displaystyle \frac{\displaystyle \partial v}{\displaystyle \partial t}
=   \nabla \cdot \left[ a_2(x)  \nabla \left( \frac{\displaystyle v(t,x)}{\displaystyle Q(x)}\right) \right]  \displaystyle + v(t,x) \left(m(x) - E_p(x) - u(t,x) - v(t,x) \right),
\\
t>0,\;x\in{\Omega},
\\
\displaystyle
\pdv{}{n}\bigg( \frac{u}{P}\bigg) = \pdv{}{n}\bigg(\frac{v}{Q} \bigg) =0,~x\in
{\partial}{\Omega}.
\end{cases}
\end{equation}  

\begin{Th}\label{perturbation_harvesting}
Suppose that the dispersal strategies are normalized as in \eqref{normal}, 
$\displaystyle \nabla \cdot \left[ a_1(x)  \nabla \left( \frac{\displaystyle m(x)}{\displaystyle P(x)}\right) \right] \not\equiv 0$ on $\Omega$, while 
$m(x)\equiv \alpha P+\beta Q$ for some $\alpha \geq 0$, $\beta>0$.

Then there exist $\varepsilon_1, \varepsilon_2>0$ such that, once either one of the following conditions holds:
\begin{description}
\item[(i)] 
$\displaystyle \int_{\Omega} E_p(x) Q(x)~dx < \varepsilon_1$,
\item[(ii)]
$|E_p(x)| < \varepsilon_2$ on $\Omega$, 
\end{description}
the second species $v$ is guaranteed to survive under the harvesting perturbation.
\end{Th}

\section{Numerical Simulations}
\label{sec:numerical}
\begin{example}
\label{ex:choose_av}

\begin{figure}[ht]
\includegraphics[width=\linewidth]{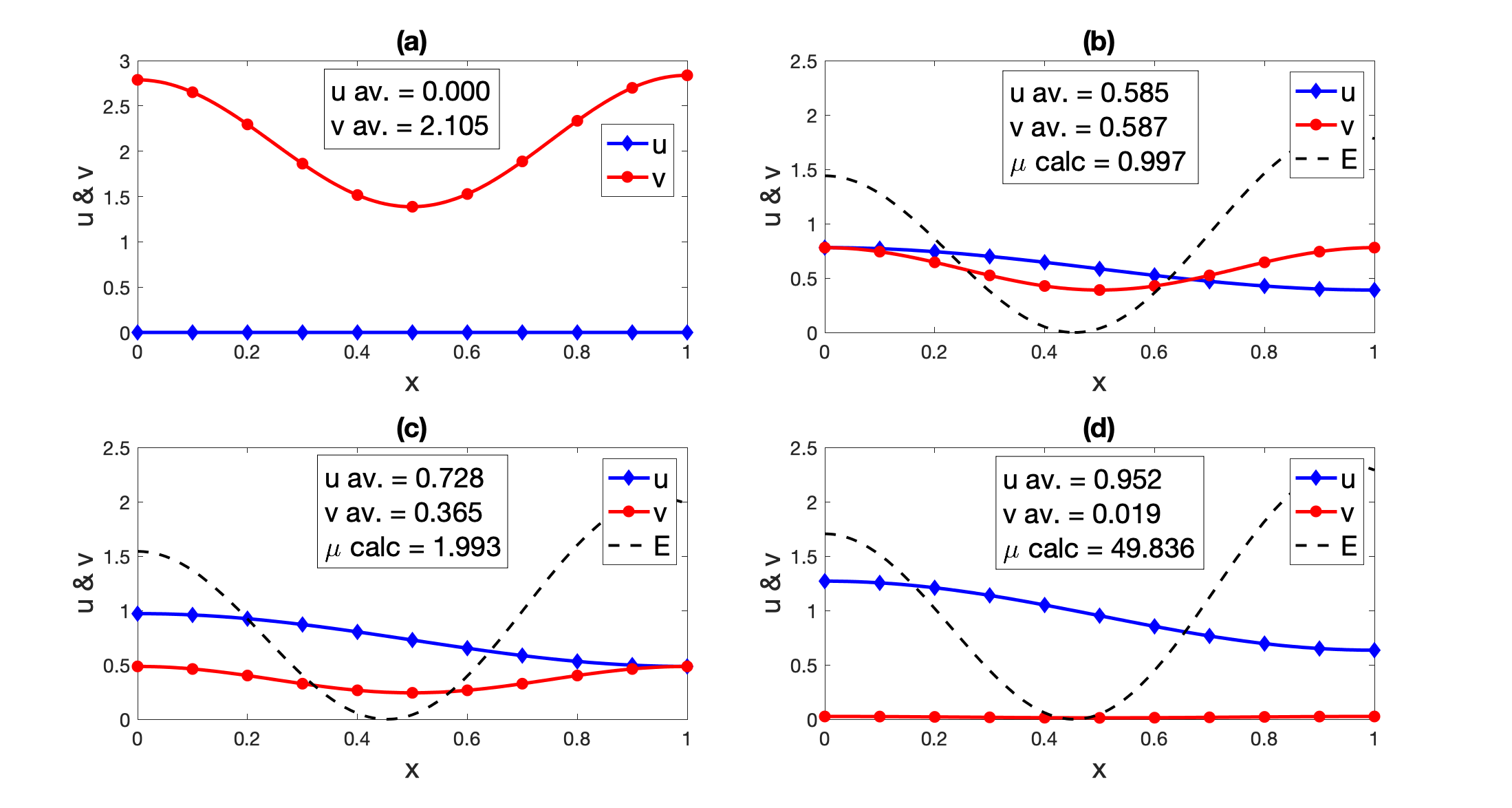}
\caption{Solutions to \eqref{system} at $t = 10 000$ for various choices of harvesting. Harvesting (when applied) is chosen such that for the coexistence solution $(u_s,v_s)$, $\int_{\Omega} u_s dx / \int_{\Omega} v_s dx = \mu$. (a) no harvesting, (b) $\mu = 1$, (c) $\mu = 2$, (d) $\mu = 50$.}
\label{fig:trim}
\end{figure}

Consider \eqref{system}, with dispersal strategies $P =  1 + (1/3)\cos(\pi x)$, $Q = 1 + (1/3)\cos(2\pi x)$, carrying capacities $K = K_1 = K_2 =  2 + \cos(2.1\pi x)$, initial conditions and diffusion coefficients $u_0 = v_0 = a_1 = a_2 = 1$, on the domain $\Omega = [0,1]$. Fig.~\ref{fig:trim} (top left), shows that with no harvesting ($E_1 = E_2 = 0)$, the $v$ population of \eqref{system}, competitively excludes the u population as time gets large and the system approaches the equilibrium. 

In Fig.~\ref{fig:trim} (top right, bottom left, bottom right), we impose harvesting in the form 
\begin{equation*}
\label{min_harvesting}
E_1(\mu;x) = K_1 - \beta^* (\mu P + Q),\quad E_2(\mu;x) = K_2 - \beta^* (\mu P + Q),
\end{equation*}
where $\beta^* = \min(\min_{x\in\Omega}( \frac{K_1}{\mu P + Q}), \min_{x\in\Omega}( \frac{K_2}{\mu P + Q}))$, and vary the $\mu$ parameter. 
The results show that not only it is possible to trim the populations to coexistence, but by using this form of harvesting, we can also force the ratios of the averages of the populations over $\Omega$ to be approximately equal to $\mu$ as predicted by Theorem \ref{stability_coexistence_2}. Here $u_{av}, v_{av}$ represent $\int_{\Omega} u_s dx, \int_{\Omega} v_s dx$ respectively. Note that because the carrying capacities $K_1 \equiv K_2 \equiv K$, we get the harvesting rates $E_1 \equiv  E_2 \equiv E$. 

In Fig.~\ref{fig:trim} (top right), we set $\mu = 1$, and find that eventually $u_{av} = 0.585$, while $v_{av} = 0.587$. This makes the calculated $\mu$ (i.e. the ratio of the averages) $ \mu_{calc} = 0.997 \approx \mu = 1$. Similar results are observed for $\mu = 2$ and $\mu = 50$. For $\mu = 2$ (Fig.~\ref{fig:trim}, bottom left), $u_{av} = 0.728$, and $v_{av} = 0.365$, which makes $\mu_{calc} = 1.993 \approx \mu = 2$. The case of $\mu = 50$ (Fig.~\ref{fig:trim}, bottom right), shows that the result still holds as $\mu$ gets large. In this case, $u_{av} = 0.952$ and $v_{av} = 0.019$ , which makes $\mu_{calc} = 49.836 \approx \mu = 50$.
\end{example}

\begin{example}
\label{ex:mimick}


Consider \eqref{ssystem}, where the first population $u$ chooses a random dispersal strategy $P=1$, while the second population mimics the dispersal of the first species which is $Q=u^*$. Here $m_1$ and $m_2$ both follow (slightly different) Gaussian normal distributions, where  $$\displaystyle m_1 = 0.1 + \frac{e^{-\frac{(x-0.5)^2}{2(0.1)^2}}}{\sqrt{2 \pi (0.1)^2}} < m_2 = 0.1 + \frac{3 e^{-\frac{(x-0.5)^2}{2(0.1)^2}}}{ 2\sqrt{2 \pi (0.1)^2}}.$$
Additionally, $u_0 = v_0 = a_1 = a_2 = 1$, $\Omega = [0,1]$. In Fig.~\ref{fig:mimick_normal}, we observe that since $m_2 > m_1$, as was predicted by Theorem \ref{mimic_invasion},   $v$ is able to invade and survive, while the home species $u$ goes extinct.

\begin{figure}[ht]
\centering
\includegraphics[width = 0.5\linewidth]{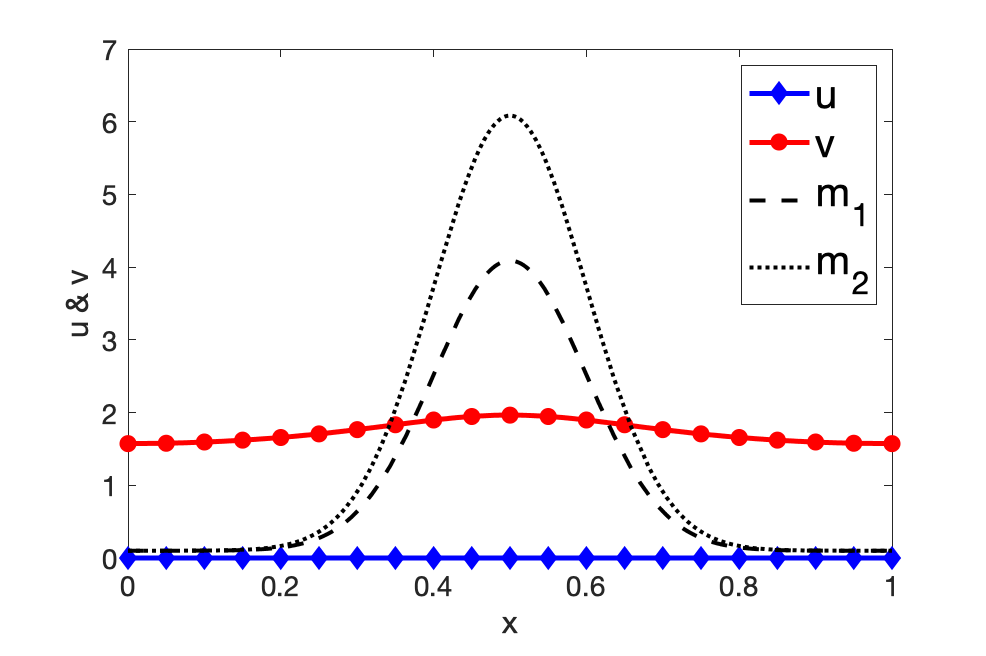}
\caption{Solutions to \eqref{ssystem} at $t=10000$, when $Q = u^*$ and $m_2 > m_1$.}
\label{fig:mimick_normal}
\end{figure}

While Theorem~\ref{mimic_invasion} and Fig.~\ref{fig:mimick_normal} demonstrate the sufficiency of the requirement that $m_2 > m_1$ on $\Omega_1 \subseteq \Omega$ and $m_2 \geq m_1$ on $\Omega$, for predicting successful invasion of $v$, it is certainly not necessary. We can still observe successful invasion of $v$ even when $m_2 = m_1$ on the whole domain, and in some cases when $m_2 < m_1$.

In Fig.~\ref{fig:success_invasion}, solutions of \eqref{ssystem} in which $v$ chooses a mimicking dispersal strategy, with $m_2 =k m_1$, $k>0, k\in\mathbb{R}$ are shown. The individual figures show solutions for various choices of $k$, where $k$ is chosen such that $m_2 \leq m_1$. Once again $m_1$ is approximately Gaussian normal, $m_1 = 0.1 + \frac{e^{-\frac{(x-0.5)^2}{2(0.1)^2}}}{\sqrt{2 \pi (0.1)^2}}$. All other parameters are the same as in Fig.~\ref{fig:mimick_normal}. In Fig.~\ref{fig:success_invasion} (top left), $k=1 \implies m_1 = m_2$, and although Theorem~\ref{mimic_invasion} does not specifically cover this case, we still observe that eventually $v$ is able to successfully invade, while $u$ goes extinct. The same phenomena is observed when $k = 0.97$ in Fig.~\ref{fig:success_invasion} (top right). Here, $m_2 < m_1$, however $v$ is still able to successfully invade, and send $u$ to extinction. When $k = 0.966$ in 
Fig.~\ref{fig:success_invasion} (bottom left), we observe that not only $v$ is able to survive, but it is actually able to coexist with $u$. However the situation changes for $k=0.96$ in Fig.~\ref{fig:success_invasion} (bottom right). In this case, even though $v$ adopts a mimicking strategy, it can no longer successfully invade. In fact, the home species $u$ survives, while $v$ goes extinct. This is due to the fact that the difference between the two carrying capacities is now too large, placing the invading species at significant disadvantage.


\begin{figure}[h!]
\centering
\includegraphics[width=\linewidth]{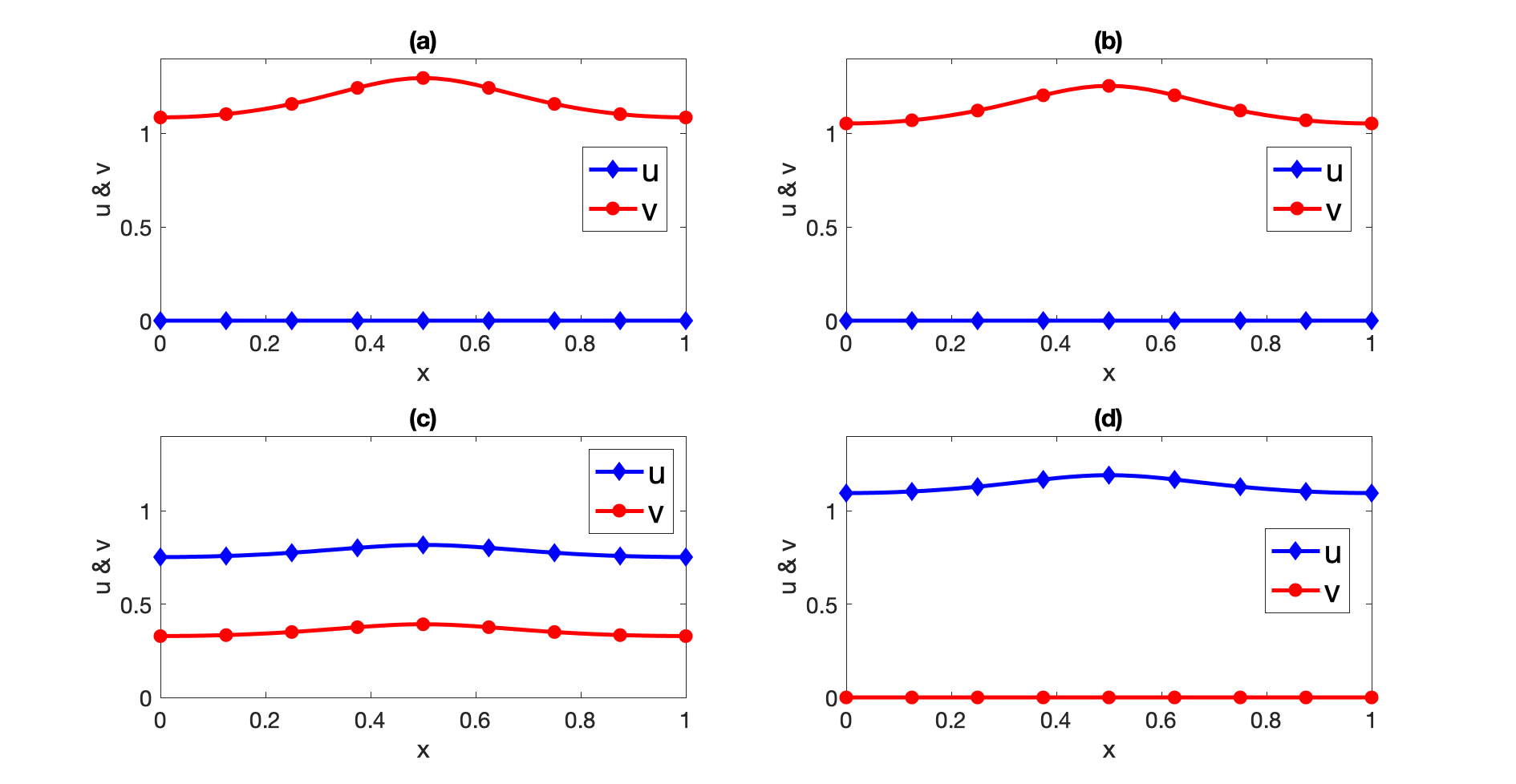}
\caption{Solutions to \eqref{ssystem} at $t = 10000$, when $Q = u^*$ and $m_2 = k m_1$. (a) $k=1$, (b) $k=0.97$, (c) $k=0.966$, (d) $k = 0.96$. Results show that in some cases when $m_2 \leq m_1$, $v$ can still successfully invade $u$ by adopting this mimicking dispersal strategy.}
\label{fig:success_invasion}
\end{figure}

\end{example}

\begin{example}
\label{ex:changing_Q1}
Consider \eqref{system2} with \eqref{Q1_form}. By Theorem \ref{stability_coexistence_3}, if \eqref{perturbation_strategy} holds, then $v$ will be able to survive even though its dispersal strategy has deviated from the ideal free pair. With $Q_1$ in the form \eqref{Q1_form}, if
\begin{equation*}
|d| \int_{\Omega} |g(x)| dx < \varepsilon_d := \displaystyle \frac{\frac{1}{\beta} \int_{\Omega} m(m-u^*)dx + \frac{\alpha}{\beta}\int_{\Omega} P(u^* - m)dx}{\| m - u^*\|_{\infty}}
\end{equation*}
then $v$ will survive.

Now, let $P =  1 + (1/3)\cos(\pi x)$, $Q = 1 + (1/3)\cos(2\pi x)$, $m = 0.5 P + 0.5 Q$, $a_1 = a_2 = u_0 = v_0 = 1$, and $\Omega = [0,1]$. Then, since $\varepsilon_d$ is dependent on $P,Q,m$, we can therefore calculate that $\varepsilon_d = 0.1577$ in this case. So to ensure survival of $v$, $|d|\int_{\Omega} |g(x)|dx < 0.1577$. 
Consider fixing the distribution of the deviance as well, choosing $\displaystyle  g(x) = \frac{\cos(\pi x)}{6} - \frac{\cos(2 \pi x)}{6}$. Then $\displaystyle \int_{\Omega} |g(x)| dx = \frac{\sqrt{3}}{4 \pi}$ and to ensure survival of $v$, we should have $ |d| < d_1$ where $ d_1 := \frac{4 \pi }{\sqrt{3}}\cdot 0.1577 \approx 1.1441$.

In Fig.~\ref{fig:changingQ1} (left) we plot the average value of solutions to \eqref{system2} at $t=10 000$, compared to increasing positive values of $d$, for the above model parameters. Note that although the averages of $P,Q,Q_1$ are all identically equal to $1$, their spatial profiles are not identical for every value of $d$ (see Fig.~\ref{fig:changingQ1}  (right)). We have chosen all the dispersal strategies to be normalized to $1$ as in \eqref{normal}, \eqref{normal_add}, so that we specifically investigate what happens when the arrangement of the resources in a dispersal strategy $Q_1$ is changed, and not what happens when the magnitude is changed.

When $d = 0$, $Q_1 = Q$, and the populations coexist, converging to the ideal free pair. For $d \in (0,1)$, coexistence is maintained, but the average value of $v$ increases, while the average value of $u$ decreases, until at $d = 1$, $Q_1 = m$, implying that $v$ will disperse exactly according to the available resources, and therefore $v$ will have the advantage over $u$ which will not be following a perfect resource distribution pattern. In this case, $v$ survives and competitively excludes $u$. When $d < d_1$, $v$ survives, in some cases coexisting with $u$, and in other cases competitively excluding $u$. This verifies that the upper bound calculated is sufficient for predicting survival of $v$. However, it is possible for $v$ to survive past that point. In fact as is seen in Fig.~\ref{fig:changingQ1}, $v$ is able to survive for all $d < d_2$. Coexistence can even be observed for $d = 2$, where $Q_1 = P $, meaning that both populations have the exact same dispersal strategies, and are therefore competing for the exact same resources. Additionally, when $d = 2$, the populations are identical, and therefore neither population has a competitive advantage. When $d > 2$, $\int_{\Omega} |m - P| dx < \int_{\Omega} |m - Q_1|dx$ and therefore, because $u$ chooses the dispersal strategy that is closest to $m$, it is going to survive 
with higher than $v$ total population size.  

\begin{figure}[h!]
    \centering
    \includegraphics[width=0.57\linewidth]{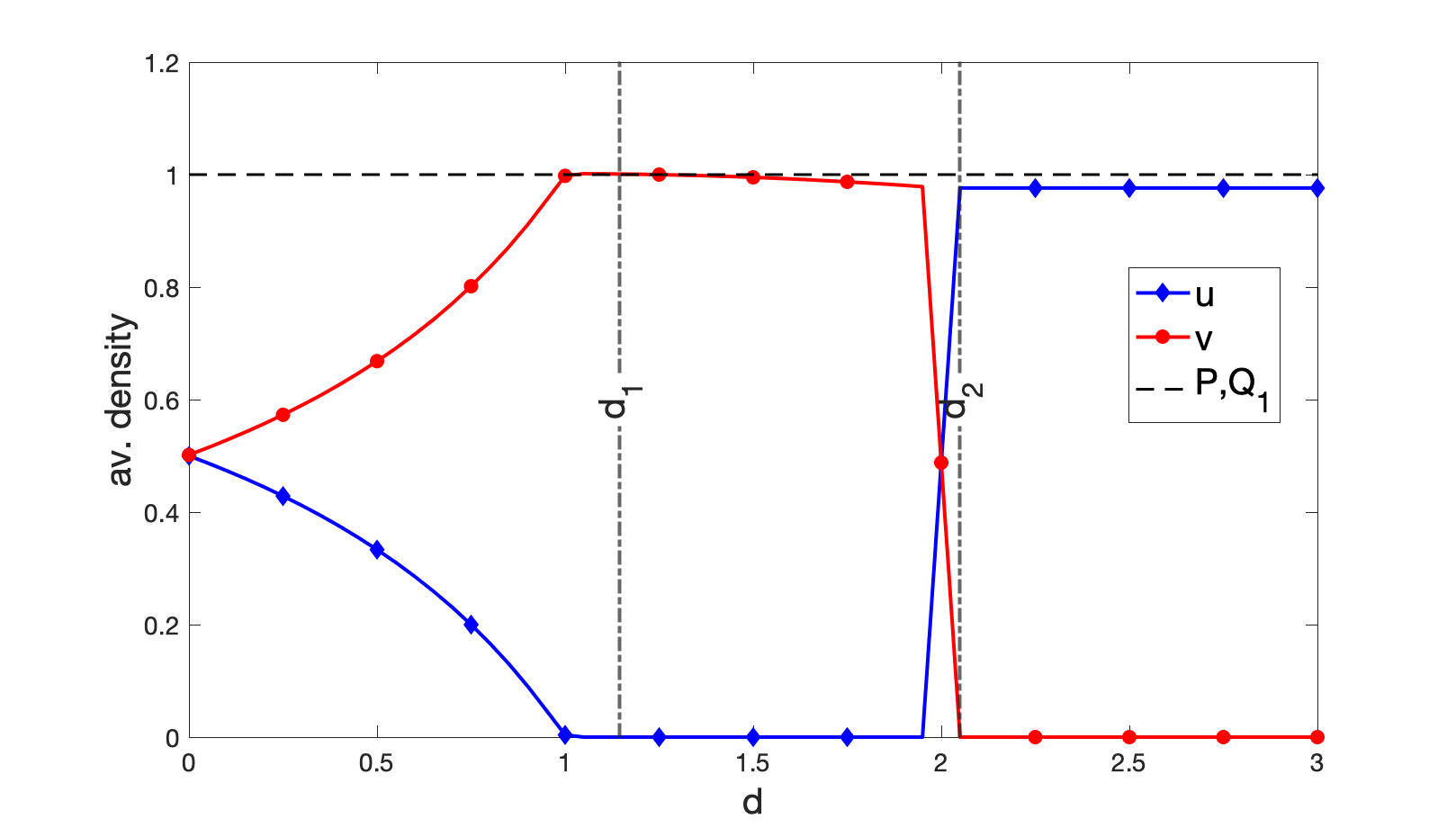}
    \hfill
    \includegraphics[width=0.42\linewidth]{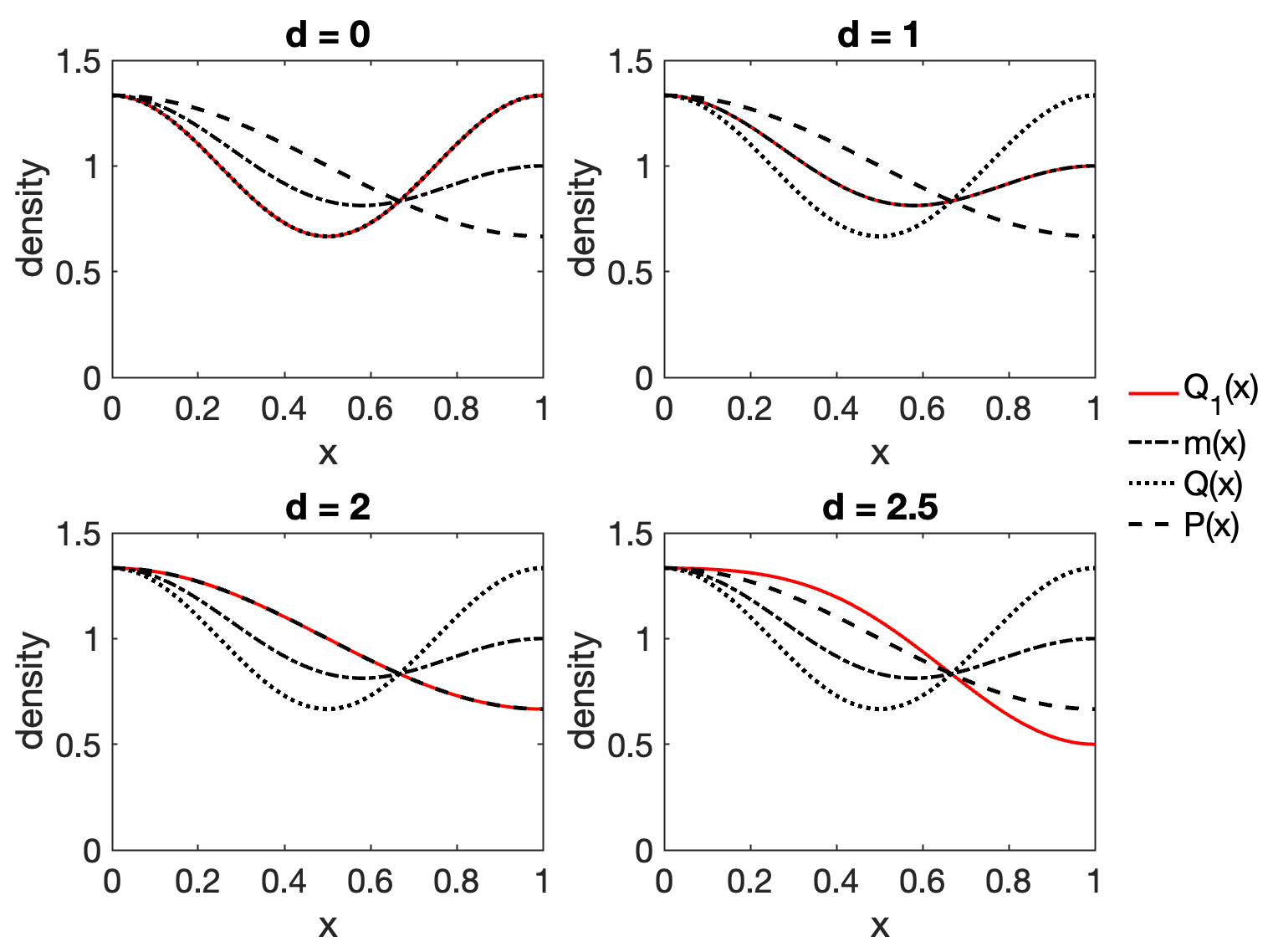}
    \caption{(Left) Average value of solutions to \eqref{system2} at $t=10{,}000$ as the deviation of $Q_1$ from $Q$ grows. When $d < d_1 = 1.1441$, then we can guarantee the survival of $v$, although it is certainly possible to have survival of $v$ past that point. Indeed $v$ can survive for all $d < d_2 = 2.05$. (Right) Spatial profiles of $P(x)$, $Q(x)$, $Q_1(x)$, and $m(x)$ for various values of $d$.}
    \label{fig:changingQ1}
\end{figure}


\end{example}

\begin{example}
\label{ex:perturb}
Consider \eqref{ssystem3} with $P =  1 + (1/3)\cos(\pi x)$, $Q = 1 + (1/3)\cos(2\pi x)$, $m = 0.5 P + 0.5 Q$, $a_1 = a_2 = u_0 = v_0 = 1$, and $\Omega = [0,1]$. Fig.~\ref{fig:perturb} displays the solutions to \eqref{ssystem3} at $t=10000$ for four different choices of harvesting perturbation $E_p$. By Theorem \ref{perturbation_harvesting} the survival of $v$ can be preserved by choosing $E_p$ such that either the following integral condition
\begin{equation}
\label{integralcond}
\int_{\Omega} E_p(x) Q(x) dx < \varepsilon_1 := \int_{\Omega} (m(x) - u^*) Q(x) dx = 0.0514
\end{equation}
holds, or 
\begin{equation}
\label{small_perturb}
|E_p(x)| < 0.0514.
\end{equation}

In Fig.~\ref{fig:perturb} (top left), $E_p = (\frac{\varepsilon_1}{2} - 0.01)\sin(\pi x)$. With this choice of $E_p$, both \eqref{small_perturb} and the integral condition \eqref{integralcond} hold. Therefore we see that the $v$ population survives. Here, the perturbation to the harvesting is small enough, that even without the moderating effect of $Q$, $v$ is still able to survive. In Fig.~\ref{fig:perturb} (top right), $E_p = 0.5 (m(x) - u^*(x))$. Since $E_p< m - u^*$, and $Q$ is normalized to $1$, it is clear that \eqref{integralcond} holds however \eqref{small_perturb} does not, as $E_p > 0.0514$ for some values of $x$. By Theorem \ref{perturbation_harvesting} the modulating effect of $Q$ is enough to ensure that the $v$ population is still able to exist, and in this case, coexist with $u$. In Fig.~\ref{fig:perturb} (bottom left), $E_p = m(x) - u^*(x)$. 
It is clear that neither \eqref{integralcond}  nor \eqref{small_perturb} are satisfied by this choice of $E_p$. As a result, $v$ becomes almost indistinguishable from $0$, achieving an average value of $1.0 \times 10^{-3}$. In this case, the perturbation to the harvesting of $v$ was enough to make the coexistence equilibrium unstable. In Fig.~\ref{fig:perturb} (bottom right), $E_p = -0.01(e^{5x} - 1)$. Once again, $E_p$ satisfies \eqref{integralcond}, and therefore $v$ survives. However, here $E_p < 0$ which corresponds to the level of harvesting of $v$ being reduced. As a result, even though $E_p$ satisfies the integral condition and $v$ survives, it does not coexist with $u$, and instead it sends $u$ to extinction.  Fig.~\ref{fig:perturb}  shows that there are multiple ways to ensure that the perturbation does not impact the survivability of $v$.
 
\begin{figure}[h!]
\centering
\includegraphics[width=\linewidth]{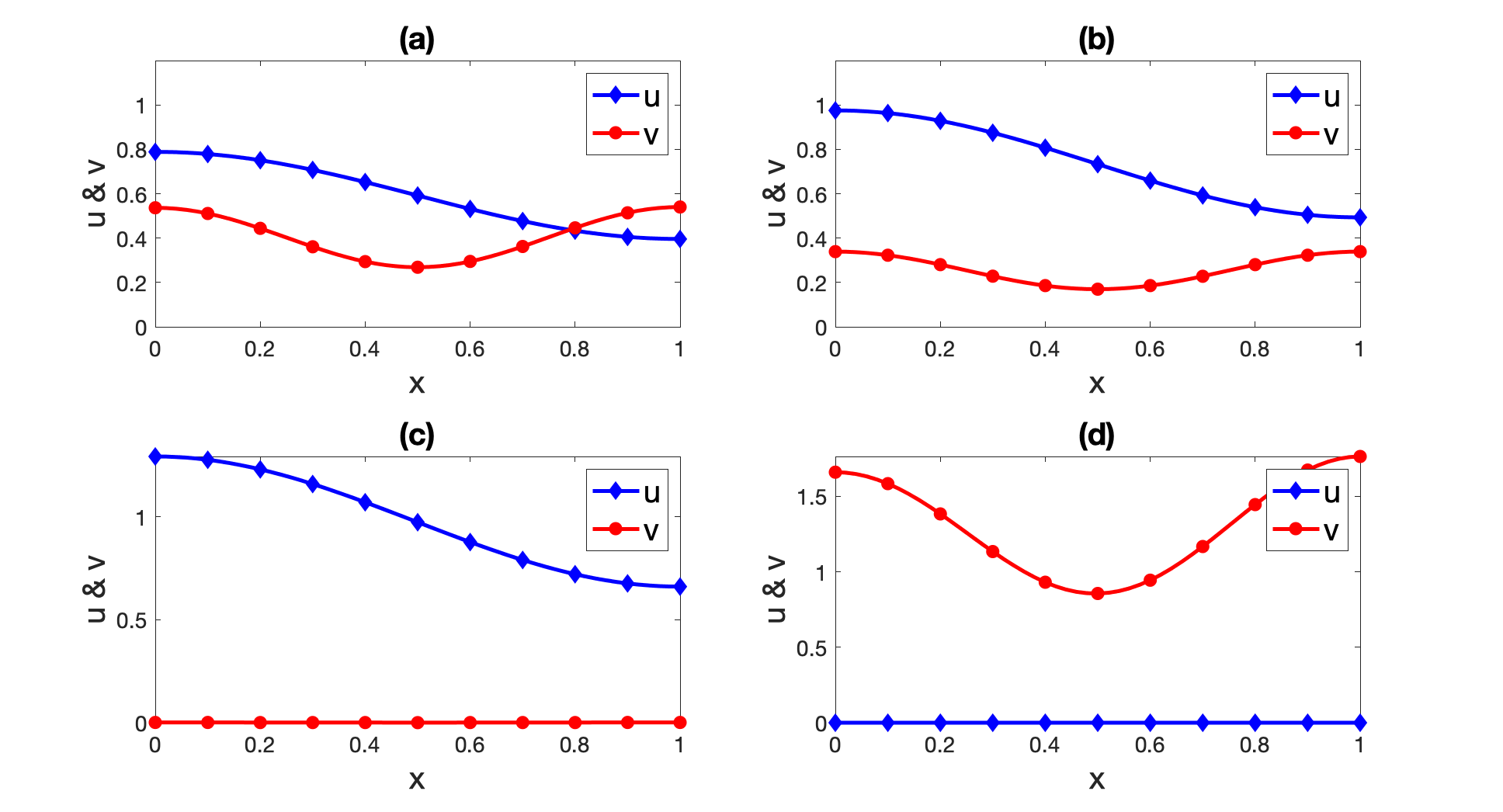}
\caption{Spatial distribution of solutions to \eqref{ssystem3} at $t=10000$ for various choices of perturbation $E_p$. (a) $E_p$ satisfies \eqref{integralcond} and \eqref{small_perturb}. (b) $E_p$ satisfies \eqref{integralcond} but not \eqref{small_perturb}. (c) $E_p$ satisfies neither \eqref{integralcond} nor \eqref{small_perturb}. (d) $E_p$ satisfies \eqref{integralcond} and $E_p<0$.}
\label{fig:perturb}
\end{figure}

\end{example}

\section{Discussion and Open Problems}
\label{sec:discussion}

This paper focused on a population's choice of dispersal strategy and the influence that choice can have on competition outcomes. It also investigated how harvesting policies can be chosen so as to promoting coexistence of the species. The main conclusions can be summarized as follows:
\begin{enumerate}
\item
For any two species which choose different diffusion strategies, it is possible to implement harvesting such that 1) the two populations coexist and 2) the ratio of their total densities can be any prescribed number. This however, does not need to be achieved by simply culling the target species aggressively. Instead, the harvesting policy can be thoughtfully chosen, such that the difference between the carrying capacity and the harvesting rate becomes adapted to the natural species distribution. For example, if two species have the same carrying capacity, then resources may be redistributed throughout the domain (through both culling and stocking in specific areas) to ensure the species coexist, and that they have the prescribed ratio of total biomass.
\item
If 1) a resident population chooses to disperse in a way that does not perfectly match the carrying capacity and 2) an invasive species chooses to disperse in such a way that it mimics the resident's spatial distribution, then it is possible for the invasion to be completely successful. In fact, we prove in Theorem \ref{mimic_invasion} that if the invasive species has an advantage in carrying capacity, then the successful invasion is guaranteed, and the resident species will be sent to extinction. If the invasive species does not have an advantage in carrying capacity, then in Example~\ref{ex:mimick} we show that it is still possible, in some circumstances, for the invader to successfully settle in the habitat.
\item
If originally, the dispersal of two species formed an ideal free pair, a slight perturbation of the dispersal strategy for one of them does not change the outcome of competition. Coexistence is preserved under small changes. However, this is not true if originally we observed one species competitively excluding the other.
\item If two populations are being harvested/or are naturally existing in such a way that they form an ideal free pair, then we state several conditions, under which perturbations in the form of harvesting will not affect the survival of the species being perturbed. We show that there are many configurations of harvesting perturbation that can preserve survival, including small perturbations, perturbations which are moderated by the dispersal strategy, and perturbations that are effectively stocking instead of harvesting.
\end{enumerate}

There are many questions arising from the current research that have not yet been answered.
\begin{enumerate}
\item
In Theorem~\ref{mimic_invasion}, if $m_1$ and $P$ are linearly independent, is the requirement that $v$ has  an advantage in the carrying capacity required to observe successful invasion? Prove or disprove the following conjecture.
\begin{enumerate}
\item
If the invasive species $v$ chooses the dispersal strategy aligned with the stationary solution $u^{\ast}$ of \eqref{eq_u} for the host species $u$,  where
$(m_1,P)$ is a linearly independent pair,
and its carrying capacity $m_2 \equiv  m_1$ on $\Omega$, this ensures successful invasion.
\item
Under the assumptions of a), the competitive exclusion of the host species $u$ is guaranteed. Note that in Example \ref{ex:mimick}, competitive exclusion of the host species was observed for identical carrying capacities.
 \end{enumerate}

\item
Prove or disprove:
\\
For any $P,Q$ in \eqref{ssystem},  where $m_1=k m_2$, once $k \geq k^{\ast} (P,Q)$,  whatever $P$ and $Q$ are,
there is competitive exclusion of the second species.
\\
If the statement above is true, is it possible to effectively evaluate $k^{\ast} (P,Q)$ ?
\item
We explored small perturbations in the diffusion and the harvesting strategies. In which cases can we state that small perturbations in the harvesting effort still keep species coexistence, when originally there is a stable coexistence equilibrium, but not necessarily an ideal free pair?  
\item
Spatial adaptation of harvesting strategies is required when certain areas are over-harvested, see \cite{CabanellasReboredoetal2017}. Can we create specific harvesting policies that could be applied in these cases? Some possible strategies could include developing a refuge outside the over-exploited area, or implementing policies which only ban harvesting in certain parts. See Example \ref{ex:perturb} for an illustration of how harvesting can be spatially perturbed to promote the survival of one species or another.

\end{enumerate}

Additionally, all of what was developed in this paper could be studied in a discrete setting, which could open up new avenues for research. For example, the directed diffusion idea has recently been extended \citep{CCZ_2022} to discrete models.

\section*{Acknowledgments}
The research of both authors was partially supported by the Natural Sciences and Engineering Research Council of Canada (NSERC). E. Braverman and J. Lawson were both partially supported by NSERC Discover Grant No. RGPIN-2020-03934. J. Lawson was supported by the NSERC Canada Graduate Scholarship - Doctoral.
The authors are very grateful to the anonymous reviewer whose thoughtful comments significantly improved presentation of our results.

\section*{Data Availability Statement}

The manuscript has no associated data.

\section*{Compliance with Ethical Standards}

The authors have no competing interests to declare that are relevant to the content of this article.
Both authors certify that they have no affiliations with or involvement in any organization or entity with any financial interest or non-financial interest in the subject matter or materials discussed in this manuscript. The research did not involve human or animal participants.
The funding information is given in Acknowledgments Section. 
Both authors contributed to the study conception and design. Problem statement - E. Braverman, analysis - both authors, simulations - J. Lawson. 
Both authors read and approved the final manuscript.

\section*{Appendix A - Proofs}

\subsection{Proof of Auxiliary results - Section~\ref{sec:auxil}}

{\em Proof of  Lemma~\ref{Lsteady2}.}   
Since $u^{\ast}>0$ and $P(x)>0$ for any $x\in \Omega$, dividing equation \eqref{eq_u} by $u^{\ast}$ and multiplying by $P$, we get
\begin{equation}\label{eq_u3}
\begin{cases}
\displaystyle \frac{\nabla \cdot [a_1 \nabla (u^{\ast}/ P)]}{(u^{\ast}/ P)} + P(x)\left(  m_1 (x)  - 
u^{\ast}(x)    \right) = 0,\; x\in\Omega,\;\\
\displaystyle \pdv{}{n} \bigg(\frac{u^*}{P} \bigg)=0,\; x\in\partial\Omega.
\end{cases}
\end{equation} 
Let us integrate \eqref{eq_u3} over $\Omega$.  Applying integration by parts and the boundary conditions in \eqref{eq_u3}, we obtain
\begin{equation*}\label{eq_u4}
\int \limits_\Omega \frac{a_1 |\nabla (u^{\ast}/P)|^2}{(u^{\ast}/P)^2}\,dx+\int \limits_\Omega P
\left(m_1-  u^{\ast}  \right)\,dx=0.
\end{equation*}
Note that $a_1$ is positive, so the first integral is non-negative leading to \eqref{eq_u2_1}. The fact that $ \nabla (u^{\ast}/P) \equiv 0$ on $\Omega$ would imply that $u^{\ast}$ is proportional to $P$.
This is only possible if  in  \eqref{eq_u} for some constant $c$ we have $m_1(x) -  cP \equiv 0$, i.e. $P$ and $m_1$ are linearly dependent, which contradicts the assumption of the lemma.

Therefore,
\begin{equation*}\label{eq_u5}
\int \limits_\Omega   P(x)\left( u^{\ast}(x)-  m_1(x)  \right)\,dx=
\int \limits_\Omega \frac{a_1(x)|\nabla (u^{\ast}/P)|^2}{(u^{\ast}/P)^2}\,dx> 0,
  \end{equation*}
which concludes the proof.

\bigskip

{\em Proof of  Lemma~\ref{Lsteady1}.}
Integrating  \eqref{eq_u} over $\Omega$ while applying the boundary conditions gives
\begin{equation*}
0=\int_{\Omega}   u^{\ast} \left( m_1-   u^{\ast}  \right)~dx=
- \int_{\Omega}  \left(   u^{\ast}- m_1 \right)^2 dx
+\int_{\Omega} m_1  \left(m_1-   u^{\ast}     \right)~dx.
\end{equation*}
Since the first integral in the right-hand side is non-positive, the second integral  satisfies
\begin{equation*}
\displaystyle \int_{\Omega} m_1  \left(m_1-   u^{\ast}     \right)~dx \geq 0.
\end{equation*}  
Moreover,  it is positive unless the first integral vanishes which happens for $u^{\ast} \equiv m_1$ only, which would imply 
$\nabla \cdot [a_1(x)\nabla (m_1(x)/P(x))] \equiv 0$ on $\Omega$. The latter identity 
contradicts the assumption of the lemma and justifies inequality \eqref{eq_est1abc}.

\bigskip

{\em Proof of  Lemma~\ref{Lsteady2abc}.}
Since $u_s(x)>0$ and $P(x)>0$ for any $x\in \Omega$, dividing the first equation in \eqref{coexist_system}  by $u_s$ and multiplying by $P$, we get
\begin{equation}\label{eq_u3_coexistence}
\begin{cases}
\displaystyle
\frac{\nabla \cdot [a_1 \nabla ( u_s /P)]}{(u_s+v_s)/ P} + P(x)\left(  m_1 (x)  -  u_s(x)  - v_s(x)     \right) = 0, \quad  x\in\Omega,  \vspace{2mm} \\
\displaystyle 
\pdv{}{n}\bigg(\frac{u_s}{P} \bigg)=0, \quad x\in\partial\Omega.
\end{cases}
\end{equation} 
Integrating  the equation in \eqref{eq_u3_coexistence} over $\Omega$ and applying integration by parts and the boundary conditions in \eqref{eq_u3_coexistence}, we obtain
\begin{equation*}
 \int \limits_\Omega \frac{a_1 |\nabla (u_s/P)|^2}{(u_s/P)^2}\,dx+\int \limits_\Omega P(x)
\left(m_1 -  u_s - v_s  \right)\,dx=0.
\end{equation*}
Hence, due to positivity of $a_1$ and the fact that $ (u_s+v_s)/P$ is non-constant on $\Omega$, we get
\begin{equation*} 
\int \limits_\Omega   P(x)\left( u_s(x) + v_s (x) -  m_1(x)  \right)\,dx=
\int \limits_\Omega \frac{a_1(x)|\nabla (u_s/P)|^2}{ ( u_s/P)^2}\,dx> 0.
  \end{equation*}

The inequality with $v^{\ast}$ is obtained similarly. We remark that the non-strict inequalities are valid, independently of the linear dependence or independence of $u_s$,$v_s$ and $P$, $Q$, respectively.

\bigskip

{\em Proof of  Lemma~\ref{Lsteady1abc}.}
Let $m_1(x) \geq m_2(x)$  hold on $\Omega$. Integrating both equations of \eqref{coexist_system} over $\Omega$ while applying the boundary conditions gives
\begin{equation}
\label{eq_lem_23_1}
\int_{\Omega}   u_s \left( m_1 - u_s - v_s \right)\, dx =0, \quad \int_{\Omega}   v_s \left( m_2 - u_s - v_s \right)\, dx =0.
\end{equation}
Note that $m_1(x) \geq m_2(x)$  leads to 
\begin{equation*}
 \int_{\Omega}   v_s \left( m_1 - u_s - v_s \right)\, dx = \int_{\Omega}   v_s \left( m_2 - u_s - v_s \right)\, dx +  \int_{\Omega}  v_s \left( m_1 -  m_2 \right)\, dx   \geq 0,
\end{equation*}
since in the second integral $v_s>0$, $m_1-m_2 \geq 0$, while the first integral is equal to zero.
Summing up the latter inequality with the first integral in \eqref{eq_lem_23_1}, we get
\begin{align*}
0 & \leq  \int_{\Omega}   (u_s + v_s)  \left( m_1 - u_s - v_s \right)\, dx \\
& =   \int_{\Omega}   (u_s + v_s - m_1)  \left( m_1 - u_s - v_s \right)\, dx +  \int_{\Omega}   m_1  \left( m_1 - u_s - v_s \right)\, dx
\\
& =  - \int_{\Omega}    \left( m_1 - u_s - v_s \right)^2 \, dx +  \int_{\Omega}   m_1  \left( m_1 - u_s - v_s \right)\, dx.
\end{align*}
As the first integral is non-positive and the sum is non-negative, the second integral is non-negative
\begin{equation*}
\int_{\Omega}   m_1  \left( m_1 - u_s - v_s \right)\, dx  \geq \int_{\Omega}    \left( m_1 - u_s - v_s \right)^2 \, dx  \geq    0.
\end{equation*}
Let us prove that a strict inequality holds. The equality will mean that $m_1 \equiv u_s + v_s$,  which contradicts the assumption of the theorem. 
Therefore  \eqref{eq_est1} is satisfied. The second part of the statement is justified similarly.

\bigskip

{\em Proof of  Lemma~\ref{semi_p1}.}
Consider the eigenvalue problem  associated with the second equation in \eqref{ssystem} 
around the equilibrium $(u^{\ast}(x),0)$   \citep{CC}  
 \begin{equation}\label{eig_p1}
\begin{cases}
\displaystyle 
\nabla \cdot \left[ a_2(x) \nabla \left(\frac{\displaystyle \psi(x) }{\displaystyle Q(x)}\right) \right]
+   \psi(x)  \left(m_2(x) -  u^{\ast}(x)  \right)=\sigma \psi(x),\; x\in \Omega,\\
\displaystyle   \pdv{}{n}\bigg( \frac{\psi}{Q}\bigg)=0,\; x\in\partial\Omega .
\end{cases}
\end{equation}
The principal eigenvalue of \eqref{eig_p1}  is defined as \citep{CC}  
\begin{equation*}
\sigma_1 =
\sup_{\psi \neq 0, \psi\in W^{1,2}} \left. \left[
-\int \limits_\Omega a_2 \left| \nabla \left( \psi/Q \right) \right|^2\,dx
+\int \limits_\Omega   \frac{\psi^2}{Q} \left( m_2 - u^{\ast}  \right)\,dx\right]
\right/
\int \limits_\Omega \frac{\psi^2}{Q}\,dx.
\end{equation*}
Let $\psi(x)= Q(x)$, we recall that 
$\displaystyle
\int \limits_\Omega  Q(x) \,dx = 1
$
by \eqref{normal}.
Substituting $\psi(x)= Q(x)$, we get that the first integral in the numerator vanishes, and
thus  the principal eigenvalue satisfies
\begin{align*}
\sigma_1 &  \geq  \int \limits_\Omega    Q(x)  \left(m_2(x)-  u^{\ast}(x)   \right) 
\,dx 
\\ & =  \int \limits_\Omega    Q(x)  \left( m_1(x)-  u^{\ast}(x)   \right) 
\,dx +  \int \limits_\Omega   Q(x)\left(  m_2 (x)-  m_1 (x)   \right)  \,dx 
\\
&  =  \frac{1}{\beta} \int \limits_\Omega (m_1(x) - \alpha P(x))\left(m_1(x)  -  u^{\ast}(x)  \right)  \,dx
 +    \int \limits_\Omega   Q(x)\left(  m_2 (x)-  m_1 (x)   \right)  \,dx    \\
& \geq  \frac{1}{\beta} \int \limits_\Omega m_1 (x)   \left(  m_1(x)  -  u^{\ast}(x)   \right) \,dx
+\frac{\alpha}{\beta} \int \limits_\Omega  P(x)\left(   u^{\ast}(x)  -  m_1(x) \right) \,dx >0.
\end{align*}
Here the first term is non-negative by \eqref{eq_est1abc} in   Lemma~\ref{Lsteady1}, while the second is positive by \eqref{eq_u2} in 
Lemma~\ref{Lsteady2}. 
Therefore, $\sigma_1$ is positive, and 
the semi-trivial steady state $(u^{\ast}(x),0)$ of \eqref{ssystem} is unstable.

\bigskip

{\em Proof of  Lemma~\ref{semi_proportional}.}
First, let us prove that  the semi-trivial steady state $(0,v^{\ast}(x))$ of \eqref{ssystem}  is unstable.
Consider the eigenvalue problem  associated with the first equation in \eqref{ssystem} 
around the equilibrium $(0,v^{\ast}(x))$  
 \begin{equation}\label{eig_p1a}
\begin{cases}
\displaystyle 
\nabla \cdot \left[ a_1(x) \nabla \left(\frac{\displaystyle \psi(x) }{\displaystyle P(x)}\right) \right]
+   \psi(x)  \left(m(x) -  v^{\ast}(x)  \right)=\sigma \psi(x),\; x\in \Omega,\\
\displaystyle   \pdv{}{n}\bigg( \frac{\psi}{P}\bigg)=0,\; x\in\partial\Omega .
\end{cases}
\end{equation}
The principal eigenvalue of \eqref{eig_p1a}  is defined as 
\begin{equation*}
\sigma_1 =
\sup_{\psi \neq 0, \psi\in W^{1,2}} \left. \left[
-\int \limits_\Omega a_1  \left| \nabla \left( \psi/P \right) \right|^2\,dx
+\int \limits_\Omega   \frac{\psi^2}{P} \left( m - v^{\ast}  \right)\,dx\right]
\right/
\int \limits_\Omega \frac{\psi^2}{P}\,dx.
\end{equation*}
Let $\psi(x)= P(x)$,  
we get that the first integral in the numerator vanishes, and
thus  the principal eigenvalue satisfies
$$
\sigma_1   \geq  \int \limits_\Omega    P(x)  \left(m(x)-  v^{\ast}(x)   \right) 
\,dx  = \frac{1}{\alpha}  \int \limits_\Omega    m(x)  \left(m(x)-  v^{\ast}(x)   \right) 
\,dx  >0
$$
 by \eqref{eq_est1abc} in   Lemma~\ref{Lsteady1}.
Therefore, $\sigma_1$ is positive, and  the semi-trivial steady state $(0,v^{\ast}(x))$ of \eqref{ssystem} is unstable.


Next, let us prove that there is no coexistence equilibrium. Assume the contrary that such equilibrium $(u_s,v_s)$ exists satisfying \eqref{coexist_system} with $m_1 \equiv m_2 \equiv m$. Adding the two equations in \eqref{coexist_system} together  and integrating using the boundary conditions leads to
\begin{equation*}
\int_{\Omega} (u_s+v_s)   \left(m- u_s - v_s \right)~dx = 0.
\end{equation*}
This is equivalent to 
\begin{equation*}
\int_{\Omega} (u_s+v_s)   \left(m- u_s - v_s \right)~dx = - \int_{\Omega} \left(m- u_s - v_s \right)^2 ~dx + \int_{\Omega} m  \left(m- u_s - v_s \right)~dx =0.
\end{equation*}
Therefore
\begin{equation}
\label{aux1}
\int_{\Omega} m  \left(m- u_s - v_s \right)~dx = \int_{\Omega} \left(m- u_s - v_s \right)^2 ~dx  > 0,
\end{equation}
unless $u_s+v_s \equiv  m$. If $u_s+v_s \equiv m$, we get two boundary value problems
\begin{equation*}
\begin{cases}
\displaystyle\nabla \cdot \left[ a_1(x) \nabla \left( \frac{u_s(x)}{P(x)} \right)  \right]
= 0,\; x\in\Omega, \\  
\displaystyle\pdv{}{n}\bigg( \frac{u_s}{P}\bigg) =  0,\; x\in\partial\Omega
\end{cases}
\end{equation*}
and
\begin{equation*}
\begin{cases}
\displaystyle \nabla \cdot \left[  a_2(x) \nabla \left( \frac{v_s(x)}{Q(x)} \right)  \right] =0
 ,\; x\in\Omega, \\   
\displaystyle  \pdv{}{n}\bigg( \frac{v_s}{Q}\bigg) = 0, \;x\in\partial\Omega
\end{cases}
\end{equation*}
The first one has solutions proportional to $P$, the second one proportional to $Q$, following the maximum principle (we can make substitutions $u_s/P = w_s, v_s/Q = z_s$ leading to constant solutions and apply the Maximum Principle to $w_s, z_s$) \citep{GilbargTrudinger}. However, as $P$ and $m$ are proportional and independent of $Q$, the only linear combination $C_1 P+ C_2 Q=m$ is possible for $C_1=\alpha$, $C_2 =0$, a contradiction to the assumption of coexistence. Thus $u_s + v_s \not\equiv m$.

Next, let $u_s+v_s \not\equiv m$. Consider the eigenvalue problem for \eqref{ssystem} associated with the first equation linearized around $(u_s, v_s)$,
 \begin{equation}
\label{add2}
\begin{cases}
\nabla \cdot \left[ a_1 \nabla \left(\frac{\displaystyle \psi(x) }{\displaystyle P(x)}\right) \right]
+   \psi(x)  \left(m -  u_s - v_s  \right)=\sigma \psi(x),\; x\in \Omega,\;\\
\displaystyle \pdv{}{n}\bigg(\frac{\psi}{P} \bigg) = 0, \;x\in\partial\Omega .
 \end{cases}
\end{equation}
Its  principal eigenvalue is  \citep{CC}
\begin{equation*}
\sigma_1 =
\sup_{\psi \neq 0, \psi\in W^{1,2}} \left. \left[
-\int \limits_\Omega a_1 |\nabla (\psi/P)|^2\,dx
+\int \limits_\Omega   \frac{\psi^2}{P} \left(m - u_s - v_s   \right)\,dx\right]
\right/
\int \limits_\Omega \frac{\psi^2}{P}\,dx.
\end{equation*}
 Substituting $\displaystyle \psi(x)= P(x) = \frac{1}{\alpha} m(x)$, where $\alpha>0$, we get by \eqref{aux1} (the integral in the denominator is equal to one),
\begin{equation}
\label{sigma}
\sigma_1  \geq \frac{1}{\alpha}  \int_{\Omega} m  \left(m- u_s - v_s \right)~dx > 0.
\end{equation}

As $(u_s,v_s)$ is an equilibrium solution, $u_s$  satisfies
\begin{equation*}
\begin{cases}
 \displaystyle  \left[ a_1 \nabla \left(\frac{\displaystyle u_s (x) }{\displaystyle P(x)}\right) \right] + u_s \left(m- u_s - v_s \right)=0, ~~x\in{\Omega},  
 \\
\displaystyle \pdv{}{n}\bigg(\frac{u_s}{P} \bigg)= 0,\; x\in
{\partial}{\Omega}
\end{cases}
\end{equation*}
and is therefore a positive principal eigenfunction of \eqref{add2} associated with the principal eigenvalue zero. 
The contradiction with \eqref{sigma}  excludes the case $u_s+v_s \not\equiv m$ as well,
which completes the proof.

\subsection{Proof of Main results - Section~\ref{sec:main}}

In this section, we provide proofs of all the main results excluding Theorem \ref{coexist_dev}, as its justification requires more effort.

\bigskip

{\em Proof of  Theorem~\ref{stability_coexistence_1}.}
All the conditions of both Lemma~\ref{semi_p1} and Lemma~\ref{semi_p2} are satisfied as  both pairs $(P,m)$ and $(Q,m)$ are linearly independent, the inequalities $m_1 \leq m_2$ and $m_2 \leq m_1$ are satisfied.
Therefore both semi-trivial equilibrium solutions are unstable. By Lemma~\ref{zero_equilibrium}, the zero equilibrium is repelling.
Thus  by Lemma~\ref{equil_charac}, the coexistence equilibrium $(\alpha P,\beta Q)$, if unique,  is a global attractor. 

It remains to prove that there are no other coexistence equilibrium solutions. We have to consider two cases: the coexistence equilibrium $(u_s,v_s)$ satisfies $u_s+v_s \equiv m$ and $u_s+v_s\not\equiv m(x)$. First, let $u_s+v_s \equiv m$. Then, denoting $w=u/P$, we get from the first equation in \eqref{ssystem} together with the boundary condition, that $w_s=u_s/P$ satisfies the following boundary value problem
\begin{equation*}
\begin{cases}
 \displaystyle  \nabla \cdot \left[ a_1(x)  \nabla w_s(x) \right] = 0, \quad x\in \Omega, \\
    \frac{\displaystyle \partial w_s}{\displaystyle \partial n}=0, \quad x\in\partial\Omega. 
 \end{cases}
\end{equation*}
This problem only has constant solutions \cite[Theorem 3.6]{GilbargTrudinger} $w_s \equiv \gamma > 0$ corresponding to $u_s = \gamma P(x)$. Similarly, $v_s = \delta Q(x)$  for some $\delta >0$. Since $u_s + v_s = m$, that implies $m(x) \equiv \gamma P(x) + \delta Q(x)$ for some  $\gamma > 0$, $\delta >0$. 
By the assumptions of the theorem $\alpha P + \beta Q \equiv m$, hence
we get 
\begin{equation*}
(\alpha-\delta) P(x) + (\gamma - \delta) Q(x) = 0 \mbox{ ~~ on ~ } \Omega.
\end{equation*}
Linear independence of $P$ and $Q$ on $\Omega$ leads to $\alpha=\gamma$, $\beta=\delta$, and shows that the only coexistence equilibrium that satisfies this case is $(u_s,v_s) = (\alpha P, \beta Q)$.

Second, let us assume that there exists another coexistence equilibrium $(u_s,v_s)$ such that $u_s(x)+v_s(x)\not\equiv m(x)$. Consider the eigenvalue problem associated with the first equation linearized around $(u_s,v_s)$
\begin{equation}
\begin{cases}
  \displaystyle  \nabla \cdot \left[ a_1(x)  \nabla \left( \frac{\displaystyle \psi(x)}{\displaystyle P(x)}\right) \right]
 \displaystyle + \psi(x) \left(m(x)-  u_s(x) - v_s(x) \right) = \sigma\psi(x), & x\in\Omega,  \\
\displaystyle \pdv{}{n}\bigg( \frac{\psi}{P}\bigg) = 0, \;x\in\partial \Omega.
\end{cases}
\label{eig2}
\end{equation}

According to the variational characterization of eigenvalues \citep{CC}, its principal eigenvalue is given by 
\begin{equation*}
\sigma_1=\sup_{\psi\neq 0,\psi\in W^{1,2}} \left[-\displaystyle \int\limits_{\Omega}  \left|  a_1 \nabla   \left( \psi/P \right) \right|^2~dx
+\int\limits_{\Omega} \frac{ \psi^2}{P}   (m(x)- u_s - v_s)\, dx\right] 
\bigg/
\displaystyle \int\limits_{\Omega} (\psi^2/P) \, dx.
\end{equation*}
Substituting $\psi(x) = P(x)$, 
we obtain
\begin{equation}
\label{inter_eig}
\begin{array}{ll}
\sigma_1 & \displaystyle  \geq \int\limits_{\Omega} P(x)  \left( m - u_s -  v_s  \right) \, dx
\\ &     \displaystyle    = \frac{1}{\alpha} \int\limits_{\Omega} \left( m - \beta Q \right)   \left( m - u_s -  v_s  \right) \, dx       
\\ &  \displaystyle    = \frac{1}{\alpha} \int\limits_{\Omega} m \left( m - u_s -  v_s  \right) \, dx   + \frac{\beta}{\alpha} \int\limits_{\Omega} Q \left( u_s +  v_s - m  \right) \, dx 
>0.
\end{array}
\end{equation}
Here the first inequality is due to the definition of the principal eigenvalue as a supremum. In the final sum, the first integral is positive
due to \eqref{eq_est1} in Lemma~\ref{Lsteady1abc} unless $u_s+v_s \equiv m$, which contradicts the assumption, and the second integral is also positive by Lemma \ref{Lsteady2abc}. However, since $(u_s,v_s)$ is an equilibrium solution of \eqref{ssystem}, $u_s,v_s$ satisfy
\begin{equation*}
\begin{cases}
 \displaystyle  \nabla \cdot \left[ a_1(x)  \nabla \left( \frac{\displaystyle u_s(x)}{\displaystyle P(x)}\right) \right]
 \displaystyle +  u_s (x) \left(m_1(x)-  u_s(x) - v_s(x) \right) = 0, ~~x\in\Omega, \\
\displaystyle \pdv{}{n}\bigg(\frac{u_s}{P} \bigg) =0,~x\in
{\partial}{\Omega}.
\end{cases}
\end{equation*}
and $u_s$ is therefore a positive principal eigenfunction of \eqref{eig2}  with the principal eigenvalue $0$. This contradiction with \eqref{inter_eig} justifies uniqueness of the coexistence equilibrium.

\bigskip

{\em Proof of  Theorem~\ref{stability_exclusion_1}.}
By Lemma~\ref{semi_proportional}, the semi-trivial equilibrium $(0,v^*)$ is unstable, and there is no coexistence equilibrium. Then, Lemma~\ref{equil_charac} implies global attractivity of the other semi-trivial equilibrium $(\alpha P, 0)$.

\bigskip

{\em Proof of  Theorem~\ref{stability_coexistence_2}.}
Let us prove the second part with a fixed $\mu>0$.  
Consider 
\begin{equation}
\label{composition}
M(x) = \frac{\mu}{\mu+1} P(x) + \frac{1}{\mu+1} Q(x) > 0, \quad x \in \Omega,
\end{equation}
then the minimum  $\zeta$  of the two continuous positive functions on the closed bounded domain $ \bar{\Omega}$   is also positive
\begin{equation}
\label{carrying}
\zeta  :=     \min\left\{   \min_{x \in \bar{\Omega}   } \frac{K_1(x)}{M(x)} ,   \min_{x \in \bar{\Omega}   } \frac{K_2(x)}{M(x)}  \right\} > 0.
\end{equation}
Choosing $\delta \in (0,1]$, denoting $\kappa = \delta \zeta $, $\zeta$ is as  in \eqref{carrying} and setting $E_1,E_2$ as
\begin{equation}
\label{harvesting_level}
E_1(x) = K_1(x)-\kappa   M(x) \geq 0, \quad  E_2(x) = K_2(x)-\kappa    M(x) \geq 0,
\end{equation}
we get system \eqref{system} which is equivalent to \eqref{ssystem} with $m(x)= \kappa  M(x)$, where $M$ is described in \eqref{composition}.
Here for $\delta = 1$ the minimal level of culling is achieved, while for $\delta \to 0$ the remaining harvested populations are infinitesimal, while relation \eqref{ratio} is preserved. Thus
\begin{equation*}
m(x) = \frac{\kappa \mu}{\mu+1} P(x) + \frac{\kappa}{\mu+1} Q(x) > 0, \quad x \in \Omega
\end{equation*}
and
\begin{equation*}
\left(u_s, v_s \right) = \left(   \frac{\kappa \mu}{\mu+1} P(x) , \frac{\kappa}{\mu+1} Q(x) \right) = (\alpha P(x), \beta Q(x)), 
~~\alpha :=  \frac{\kappa \mu}{\mu+1}, ~ \beta := \frac{\kappa}{\mu+1}
\end{equation*}
is a coexistence equilibrium of \eqref{ssystem} and thus of \eqref{system} with $E_1$, $E_2$ as described in \eqref{harvesting_level}.

By Theorem~\ref{stability_coexistence_1}, all solutions of \eqref{ssystem} (and thus of \eqref{system} )  with non-negative non-trivial initial conditions converge to the coexistence equilibrium $(\alpha P,\beta Q)$. Further, by the definition of $\alpha$ and $\beta$, we have 
$$ \alpha = \mu \beta, ~~  \int_{\Omega} u_s \, dx = \int_{\Omega} \alpha P \, dx = \mu \beta \int_{\Omega} P\, dx
= \mu \beta \int_{\Omega} Q\, dx 
= \mu \int_{\Omega} v_s \, dx$$
due to normalization assumption \eqref{normal}, which concludes the proof.

\bigskip

{\em Proof of  Theorem~\ref{mimic_invasion}.}
Note that under the assumptions of the theorem, \eqref{ssystem} takes the form
\begin{align}
\label{ssystem_2}
\begin{cases}
\displaystyle \frac{\displaystyle \partial u}{\displaystyle \partial t}
=   \displaystyle  \nabla \cdot \left[ a_1(x)  \nabla \left( \frac{\displaystyle u(t,x)}{\displaystyle P(x)}\right) \right]
 \displaystyle + u(t,x) \left(m_1(x)-  u(t,x) - v(t,x) \right), \vspace{1mm}
\\
\displaystyle \frac{\displaystyle \partial v}{\displaystyle \partial t}
=   \nabla \cdot \left[ a_2(x)  \nabla \left( \frac{\displaystyle v(t,x)}{\displaystyle u^{\ast}(x)}\right) \right]  \displaystyle + v(t,x) \left(m_2(x) - u(t,x) - v(t,x) \right),
\\
t>0,\;x\in{\Omega},\quad m_2(x) \geq m_1(x),~~x\in{\Omega},~~ m_2(x) > m_1(x),~~x\in{\Omega_1} \subseteq \Omega,
\\
\displaystyle \frac{\partial}{\partial n}\bigg( \frac{u}{P}\bigg) = \frac{\partial}{\partial n} \bigg( \frac{v}{u^*}\bigg) = 0, \quad x\in \partial \Omega.
\end{cases}
\end{align}

In order to guarantee successful invasion, we have to show that the competitive exclusion equilibrium, $(u^{\ast}(x),0)$, is unstable. Let us refer to the second equation in \eqref{ssystem_2}. Consider the eigenvalue problem associated with the second equation in \eqref{ssystem_2} around the semi-trivial equilibrium $(u^{\ast}(x),0)$  which is 
 \begin{equation}\label{eig_p1_invade}
\begin{cases}
\displaystyle
\nabla \cdot \left[ a_2(x) \nabla \left(\frac{\displaystyle \psi(x) }{\displaystyle u^{\ast}(x)}\right) \right]
+   \psi(x)  \left(m_2(x) -  u^{\ast}(x)  \right)=\sigma \psi(x),  \;  x\in \Omega, \\
 \displaystyle \pdv{}{n}\bigg(\frac{\psi}{u^*} \bigg)=0,  \;  x\in\partial\Omega.
 \end{cases}
\end{equation}
The principal eigenvalue of \eqref{eig_p1_invade} is defined as \citep{CC}
$$
\sigma_1 =
\sup_{\psi \neq 0, \psi\in W^{1,2}} \left. \left[
-\int \limits_\Omega a_2 |\nabla (\psi/u^{\ast} )  |^2\,dx
+\int \limits_\Omega   \frac{\psi^2}{u^{\ast}} \left(m_2 - u^{\ast}  \right)\,dx\right]
\right/
\int \limits_\Omega \frac{\psi^2}{u^{\ast}}\,dx.
$$
Substituting $ \psi(x) = u^{\ast}(x)$, we get that the principal eigenvalue is not less than
\begin{align*}
\sigma_1 &  \geq \left. \left[  0 +  \int \limits_\Omega     u^{\ast}(x)  \left(m_2(x)-  u^{\ast}(x)   \right) 
\,dx \right] \right/
\int \limits_\Omega u^{\ast}  \,dx
\\ & =   \left. \left[  \int \limits_\Omega    u^{\ast}(x)  \left( m_1(x)-  u^{\ast}(x)   \right) 
\,dx +     \int \limits_\Omega    u^{\ast}(x) \left(  m_2 (x)-  m_1 (x)   \right)  \,dx   \right] \right/
\int \limits_\Omega u^{\ast}  \,dx
\\
&=   \left[   \int \limits_\Omega    u^{\ast}(x)  \left( m_1(x)-  u^{\ast}(x)   \right) 
\,dx +     \int \limits_{\Omega_1}   u^{\ast}(x) \left(  m_2 (x)-  m_1 (x)   \right)  \,dx  \right. \\
& ~~~+      \left.  \left. \int \limits_{\Omega \backslash \Omega_1}    u^{\ast}(x) \left(  m_2 (x)-  m_1 (x)   \right)  \,dx  \right] \right/
\int \limits_\Omega u^{\ast}  \,dx.
\end{align*}
Integrating \eqref{eq_u} and applying the corresponding boundary conditions, it can be shown that the first integral in the last line is equal to zero. Further, since $m_2>m_1$ on $\Omega_1$, and the integral of a positive continuous function over a domain with positive area is positive, thus we find that the second integral in the last line is positive. Finally, the third integral is non-negative, since $u^*>0$ and $m_2\geq m_1$ on $\Omega \backslash \Omega_1$. Thus overall we see that $\sigma_1$ is positive, and the semi-trivial steady state $(u^{\ast}(x),0)$ of \eqref{ssystem_2} is unstable. Because $(u^*,0)$ is unstable, the invasion of the second species will be successful, either coexisting with the resident species, or bringing the resident to extinction.

\bigskip

{\em Proof of  Lemma~\ref{semi_perturbation}.}
We have to prove that for a certain $\varepsilon$, the semi-trivial equilibrium $(u^{\ast}(x),0)$ of \eqref{system2} is unstable. 
To this end, consider the eigenvalue problem associated with the second equation in \eqref{system2} linearized around the equilibrium $(u^{\ast}(x),0)$ 
 \begin{equation}\label{eig_p1_add}
\begin{cases}
\displaystyle 
\nabla \cdot \left[ a_2(x) \nabla \left(\frac{\displaystyle \psi(x) }{\displaystyle Q_1(x)}\right) \right]
+   \psi(x)  \left(m(x) -  u^{\ast}(x)  \right)=\sigma \psi(x), \; x\in \Omega,\\ 
\displaystyle \pdv{}{n} \bigg( \frac{\psi}{Q_1}\bigg)=0, \;  x\in\partial\Omega.
 \end{cases}
\end{equation}
The principal eigenvalue of \eqref{eig_p1_add}  is  
\begin{equation*}
\sigma_1 =
\sup_{\psi \neq 0, \psi\in W^{1,2}} \left. \left[
-\int \limits_\Omega a_2 |\nabla (\psi/Q_1)|^2\,dx
+\int \limits_\Omega   \frac{\psi^2}{Q_1} \left(m - u^{\ast}  \right)\,dx\right]
\right/
\int \limits_\Omega \frac{\psi^2}{Q_1}\,dx
\end{equation*}
and choosing $\psi(x)=  Q_1(x)$, we obtain
\begin{equation*}
\begin{aligned}
\sigma_1 &\geq  \int \limits_\Omega    Q_1(x)  \left(m(x)-  u^{\ast}(x)   \right) \,dx\\
&=  \int \limits_\Omega    (Q_1(x)-Q(x)) \left(m(x)-  u^{\ast}(x)   \right) \,dx
+ \int \limits_\Omega    Q(x)  \left(m(x)-  u^{\ast}(x)   \right) \,dx\\
&  =  \int \limits_\Omega    (Q_1(x)-Q(x)) \left(m(x)-  u^{\ast}(x)   \right) \,dx + 
\frac{1}{\beta} \int \limits_\Omega m (x)\left(  m(x)  -  u^{\ast}(x)   \right) \,dx \\
& +\frac{\alpha}{\beta} \int \limits_\Omega  P(x)\left(   u^{\ast}(x)  -  m(x) \right) \,dx.\\
\end{aligned}
\end{equation*}
Following the proof of Lemma~\ref{semi_p1}, we get that the second term is positive by Lemma~\ref{Lsteady1}, and the third one by Lemma~\ref{Lsteady2}. 
Denote
\begin{equation*}
\displaystyle \gamma := \frac{1}{\beta} \int \limits_\Omega m (x)\left(  m(x)  -  u^{\ast}(x)   \right) \,dx
+\frac{\alpha}{\beta} \int \limits_\Omega  P(x)\left(   u^{\ast}(x)  -  m(x) \right) \,dx >0 .
\end{equation*}
When $m\neq u^*$, choose $Q_1$ such that $\displaystyle \int \limits_\Omega    |Q_1(x)-Q(x)|  \,dx < \varepsilon < \varepsilon_d$, where
\begin{equation}
\label{eps_d}
\varepsilon_d := \displaystyle \frac{\gamma}{\displaystyle \max_{x\in\Omega} |m(x) - u^*(x)|},
\end{equation}
so that
\begin{equation*}
\begin{aligned}
\sigma_1 &\geq  \int \limits_\Omega    (Q_1(x)-Q(x)) \left(m(x)-  u^{\ast}(x)   \right) \,dx + \gamma \\
&\geq  -  \max_{x \in \Omega}  \left| m(x) -  u^{\ast}(x)  \right|
\int \limits_\Omega  \left| Q(x)-Q_1(x) \right| ~dx + \gamma\\
&>  -  \max_{x \in \Omega}  \left| m(x) -  u^{\ast}(x)  \right|
\varepsilon_d + \gamma =0.
\end{aligned}
\end{equation*}
Thus the principal eigenvalue of \eqref{eig_p1_add} is positive, and the equilibrium $(u^*,0)$  of \eqref{system2} is unstable, which implies at least survival of the $v$ population.

Finally, let us justify that $m \equiv  u^{\ast}$  is impossible. Substituting $m=\alpha P + \beta Q$ in the first equation of \eqref{system2} and taking into account the boundary conditions, we get 
$\displaystyle \nabla \cdot \left[ a_1(x)  \nabla \left( \frac{\displaystyle m(x)}{\displaystyle P(x)}\right) \right] \equiv 0$ on $\Omega$, which contradicts the assumption of the lemma.

\bigskip

{\em Proof of  Theorem~\ref{stability_coexistence_3}.}
\noindent(1) The first item immediately follows from Lemma~\ref{semi_perturbation}, where in the proof this $\varepsilon_d$ is evaluated in \eqref{eps_d}.\\
\noindent(2)  By Theorem~\ref{stability_exclusion_1}, if $Q_1(x)=\alpha m(x)$ for some $\alpha>0$ and all $x \in \Omega$, the second species in \eqref{system2} brings the first one to extinction. We can evaluate the distance between the current and the winning strategy for the second species. With the maximum-norm in $\Omega$, we compute 
$$
d_{\max} := \inf_{\alpha > 0} \max_{x \in \bar{\Omega}} \left\{ |Q(x) - \alpha m(x) | \right\}.  
$$
Allowing $|  Q(x) - Q_1(x)| \geq d_{\max}$  for some $x \in \bar{\Omega}$, while choosing $Q_1= \alpha m(x)$, competitive exclusion is achieved.

\bigskip

{\em Proof of  Lemma~\ref{lemma_new_deviation}.}
We have to prove that for a certain $\varepsilon$, the semi-trivial equilibrium $(0,v_1^{\ast}(x))$ of \eqref{system2} is unstable. 
The eigenvalue problem  associated with the first equation in \eqref{system2}
around the equilibrium $(0,v_1^{\ast}(x))$  is
 \begin{equation}\label{eig_p1_add_2}
\begin{cases}
\displaystyle 
\nabla \cdot \left[ a_1(x) \nabla \left(\frac{\displaystyle \psi(x) }{\displaystyle  P(x)}\right) \right]
+   \psi(x)  \left(m(x) -  v_1^{\ast}(x)  \right)=\sigma \psi(x), & x\in \Omega,\\ \displaystyle 
\pdv{}{n}\bigg( \frac{\psi}{P}\bigg) = 0, \quad  x\in\partial\Omega.
 \end{cases}
\end{equation}
The principal eigenvalue of \eqref{eig_p1_add_2}  equals   
\begin{equation*}
\sigma_1 =
\sup_{\psi \neq 0, \psi\in W^{1,2}} \left. \left[
-\int \limits_\Omega a_1 |\nabla (\psi/P)|^2\,dx
+\int \limits_\Omega   \frac{\psi^2}{P} \left(m - v_1^{\ast}  \right)\,dx\right]
\right/
\int \limits_\Omega \frac{\psi^2}{P}\,dx.
\end{equation*}
Choosing $\psi(x)=  P(x)$,  we get
\begin{equation*}
\begin{aligned}
\sigma_1 &\geq  \int \limits_\Omega    P(x)  \left(m(x)-  v_1^{\ast}(x)   \right) \,dx  \\
&=  \frac{1}{\alpha} \int \limits_\Omega    (m(x) - \beta Q(x)) \left(m(x)-  v_1^{\ast}(x)   \right) \,dx
\\ &=  \frac{1}{\alpha}  \int \limits_\Omega  (m(x) - \beta Q)(m(x) - v^*(x))  \,dx +  \frac{1}{\alpha}  \int \limits_\Omega  (m(x) - \beta Q(x))(v(x) - v_1^*(x)) \,dx 
\\  &=    \frac{1}{\alpha}  \int \limits_\Omega  m(x)(m(x) - v^*(x))  \,dx  +  \frac{\beta}{\alpha}  \int \limits_\Omega  Q(x) (v^*(x)-m(x))  \,dx \\ & +   
 \frac{1}{\alpha}  \int \limits_\Omega  (m(x) - \beta Q(x))(v(x) - v_1^*(x)) \,dx.
\end{aligned}
\end{equation*}
Denote
\begin{equation*}
\delta :=    \int \limits_\Omega  m(x)(m(x) - v^*(x))  \,dx  +   \beta \int \limits_\Omega  Q(x) (v^*(x)-m(x))  \,dx > 0,
\end{equation*}
as the first term is positive by Lemma~\ref{Lsteady1}, and the second term by 
Lemma~\ref{Lsteady2}.
Choose
\begin{equation*}
\varepsilon = \sup_{x \in \Omega} \left| v^*(x)-v_1^*(x) \right|  < \delta \left(  \left|  \int \limits_\Omega  (m(x) - \beta Q(x))\, dx  \right| \right)^{-1},
\end{equation*}
if the integral is not zero, we get $\sigma_1>0$ for the principal eigenvalue of \eqref{eig_p1_add_2}. If the integral is zero, $\sigma_1>0$ as a sum of two positive integrals.
Thus the first species survives, which completes the proof.

\bigskip

{\em Proof of  Theorem~\ref{semi_perturbation_aligned}.}
We have to prove that for a certain $\varepsilon$, the semi-trivial equilibrium $(u^{\ast}(x),0)$ of \eqref{system2} is unstable. To this end, consider the eigenvalue problem associated with the second equation in \eqref{system2} around the equilibrium $(u^{\ast}(x),0)$ 
 \begin{equation}\label{eig_p1_add_4}
\begin{cases}
\displaystyle 
\nabla \cdot \left[ a_2(x) \nabla \left(\frac{\displaystyle \psi(x) }{\displaystyle Q_1(x)}\right) \right]
+   \psi(x)  \left(m_2(x) -  u^{\ast}(x)  \right)=\sigma \psi(x), & x\in \Omega,\\ 
\displaystyle 
 \pdv{}{n}\bigg( \frac{\psi}{Q_1}\bigg)=0, \quad  x\in\partial\Omega.
\end{cases}
\end{equation}
The principal eigenvalue of \eqref{eig_p1_add_4}  is 
\begin{equation*}
\sigma_1 =
\sup_{\psi \neq 0, \psi\in W^{1,2}} \left. \left[
-\int \limits_\Omega a_2 |\nabla (\psi/Q_1)|^2\,dx
+\int \limits_\Omega   \frac{\psi^2}{Q_1} \left(m - u^{\ast}  \right)\,dx\right]
\right/
\int \limits_\Omega \frac{\psi^2}{Q_1}\,dx,
\end{equation*}
and choosing $\psi(x) = \beta Q_1(x) > 0$,  we get
\begin{equation*}
\begin{aligned}
\sigma_1 \geq& \frac{1}{\beta^2} \int \limits_\Omega  \beta^2   Q_1(x)  \left(m(x)-  u^{\ast}(x)   \right) \,dx \\ 
= & \int \limits_\Omega    (Q_1(x) -  Q(x)) \left(m(x)-  u^{\ast}(x)   \right) \,dx
+ \int \limits_\Omega    Q(x)  \left(m(x)-  u^{\ast}(x)   \right) \,dx
\\ 
= &  \int \limits_\Omega   (Q_1(x)-Q(x)) \left(m(x)-  u^{\ast}(x)   \right) \,dx + \frac{1}{\beta} \int \limits_\Omega  m(x) \left(m(x)-  u^{\ast}(x)   \right) \,dx .
\end{aligned}
\end{equation*}
However, according to \eqref{eq_est1abc} in Lemma~\ref{Lsteady1},
\begin{equation*}
A:= \int\limits_{\Omega}m \left( m(x) - u^{\ast}(x) \right)  \,dx   > 0.
\end{equation*}
Notice that $u^{\ast}(x)   \not\equiv m(x)$, and denote
\begin{equation*}
\varepsilon_0  := \frac{A}{\beta }  \left(  \max_{x \in \Omega} \left|   m(x) -  u^{\ast}(x) \right|   \right)^{-1}.
\end{equation*}
Following the same logic as in the proof of Lemma \ref{semi_perturbation} and choosing $Q_1$ such that $\varepsilon < \varepsilon_0$, we get a positive principal eigenvalue. Thus, the second species sustains, either in combination with its competitor or excluding it, which completes the proof.

\bigskip

{\em Proof of  Theorem~\ref{perturbation_harvesting}.}
Similarly to \eqref{eig_p1}, we consider the eigenvalue problem associated with the second equation in \eqref{ssystem3} 
linearized around the equilibrium $(u^{\ast}(x),0)$ 
 \begin{equation*}
\begin{cases}
\displaystyle 
\nabla \cdot \left[ a_2(x) \nabla \left(\frac{\displaystyle \psi(x) }{\displaystyle Q(x)}\right) \right]
+   \psi(x)  \left(m(x) - E_p(x) -  u^{\ast}(x)  \right)=\sigma \psi(x),\; x\in \Omega,\\
\displaystyle   \pdv{}{n} \bigg(\frac{\psi}{Q} \bigg)=0,\; x\in\partial\Omega,
\end{cases}
\end{equation*}
where the principal eigenvalue of \eqref{eig_p1}  is 
\begin{equation*}
\sigma_1 =
\sup_{\psi \neq 0, \psi\in W^{1,2}} \left. \left[
-\int \limits_\Omega a_2 \left| \nabla \left( \frac{\psi}{Q} \right) \right|^2\,dx
+\int \limits_\Omega   \frac{\psi^2}{Q} \left( m - E_p - u^{\ast}  \right)\,dx\right]
\right/
\int \limits_\Omega \frac{\psi^2}{Q}\,dx.
\end{equation*}
Let us remark that the assumptions of the theorem guarantee $u^{\ast} \not\equiv m$, and the constant
\begin{align*}
\varepsilon_1  & := \frac{1}{\beta} \int \limits_\Omega (m(x) - \alpha P(x))\left(m(x) -  u^{\ast}(x)  \right)  \,dx \\
& =
\frac{1}{\beta} \int \limits_\Omega m(x)   \left(  m(x)  -  u^{\ast}(x)   \right) \,dx
+\frac{\alpha}{\beta} \int \limits_\Omega  P(x)\left(   u^{\ast}(x)  -  m(x) \right) \,dx >0,
\end{align*}
since the first term is positive  by \eqref{eq_est1abc} in   Lemma~\ref{Lsteady1}, while the second is positive due to \eqref{eq_u2_1} in 
Lemma~\ref{Lsteady2}. Substituting $\psi(x)= Q(x)$, we obtain that the supremum is not less than
\begin{align*}
\sigma_1 &  \geq  \int \limits_\Omega    Q(x)  \left(m(x) - E_p(x) -  u^{\ast}(x)   \right) 
\,dx 
\\ & =  \int \limits_\Omega    Q(x)  \left( m(x)-  u^{\ast}(x)   \right) 
\,dx -   \frac{1}{\beta} \int \limits_\Omega   \beta Q(x) E_p(x)  \,dx 
\\
&  =  \frac{1}{\beta} \int \limits_\Omega (m(x) - \alpha P(x))\left(m(x) -  u^{\ast}(x)  \right)  \,dx
 -    \int \limits_\Omega   Q(x) E_p(x)   \,dx    \\
&
 =  \varepsilon_1 -   \int \limits_\Omega   Q(x) E_p(x)   \,dx  >0,
\end{align*}
once $\displaystyle \int_{\Omega} E_p(x) Q(x)~dx < \varepsilon_1$, which concludes the proof of (i). Assuming $\varepsilon_2=\varepsilon_1$, due to normalization of $Q$ in $\Omega$, we get (ii), completing the proof.

\section*{Appendix B: Justification of Theorem \ref{coexist_dev}}


\subsection{Auxiliary results}

First, we state some auxiliary results. Recall that by \eqref{Q1_form}, we assume $Q_1 = Q + d g(x)$, where $$\int_{\Omega} g(x) dx = 0.$$

Let $z_1^*$ be the solution of
\begin{equation}
\label{z_1}
    \begin{cases}
        \displaystyle
        \nabla \cdot (a_2 \nabla z_1^*) + Q_1 z_1^* (m - Q_1 z_1^*) = 0,\quad x\in \Omega, \\
        \displaystyle
        \pdv{z_1^*}{n} = 0, \quad \partial \Omega.
    \end{cases}
\end{equation}
$z_1^*$ is 
connected to $v_1^*$ by the relation $v_1^* = Q_1 z_1^*$. 
When $d = 0$, it is clear that $Q_1 = Q$. We will also need the solution of \eqref{z_1} when $d=0$. 
Let $z^*$ be   the solution of
\begin{equation}
\label{zstar}
    \begin{cases}
        \displaystyle
        \nabla \cdot (a_2 \nabla z^*) + Q z^* (m - Q z^*) = 0,\quad x\in \Omega\\
        \displaystyle
        \pdv{z^*}{n} = 0, \quad \partial \Omega.
    \end{cases}
\end{equation}

\begin{Lemma}
\label{z_star_bound}
Let $m = \alpha P + \beta Q$, $m > 0$, $P,Q>0$, and 
\begin{equation}
\label{cond:lemma1}
    \displaystyle \frac{\sup_{x\in\Omega} \frac{m}{Q}}{\inf_{x\in\Omega} \frac{m}{Q}} <2.
\end{equation}
Then $z^* > m/(2Q)$.
\end{Lemma}

\begin{proof}
Let
    \begin{equation*}
        \gamma = \inf_{x\in\Omega} \frac{m}{Q}.
    \end{equation*}
Notice that since \eqref{cond:lemma1} holds, we have
\begin{equation*}
\frac{m}{2Q} \leq \sup_{x\in\Omega} \frac{m}{2Q} < \inf_{x\in\Omega} \frac{m}{Q} = \gamma \leq \frac{m}{Q}.
\end{equation*} 
Let 
\begin{equation*}f(x,y) = Q(x) y (m(x) - Q(x) y).
\end{equation*}  
It is clear that $f(x,y) \geq 0$ for  
$y\in [0,\frac{m}{Q}]$. 
Thus, $f(x,\gamma) \geq 0$.

Let $\hat{z} = \gamma$,
\begin{equation*}
    \begin{aligned}
        &\nabla \cdot (a_2 \nabla \hat{z}) + Q \hat{z} (m - Q \hat{z})\\
        &= \nabla \cdot (a_2 \nabla \gamma) + Q \gamma (m - Q \gamma)\\
        &= Q \gamma (m - Q \gamma)\quad \text{since }\gamma\text{ is constant}\\
        & = f(x,\gamma)\\
        & \geq 0
    \end{aligned}
\end{equation*}
Thus $\hat{z} = \gamma$ is a lower solution of the equilibrium problem \eqref{zstar} \citep{Pao}. Since $\hat{z} = \gamma$ is a lower solution, $z^* \geq \gamma$, and since $\gamma > \frac{m}{2Q}$, we therefore have that $z^* > \frac{m}{2Q}$ as required.
\end{proof}

\begin{Lemma}
\label{expand_z}
    Let $Q_1 $ be in the form \eqref{Q1_form}, and let \eqref{cond:lemma1} hold. Then for small $d>0$, the solution of \eqref{z_1} can be expanded as
    \begin{equation}
    \label{z_expand}
        z_1^* = z^* + f_1(x)d + f_2(x) d^2 + f_3(x)  d^3,
    \end{equation}
    with $f_1(x)$ the solution of 
    \begin{equation}
    \begin{cases}
    \displaystyle
        \nabla \cdot (a_2 \nabla f_1(x)) + m Q f_1(x) - 2 f_1(x) Q^2 z^* + g m z^* - 2 g Q (z^*)^2=0, \quad x\in \Omega, \\
        \displaystyle
        \pdv{f_1(x)}{n} = 0,\quad x \in \partial\Omega,
    \end{cases}
    \end{equation}
    $f_2(x)$ the solution of
    \begin{equation}
    \begin{cases}
    \displaystyle
        \nabla \cdot (a_2 \nabla f_2(x)) + f_1(x) (g m - f_1(x) Q^2 - 4gQz^*) + f_2(x) Q (m - 2Q z^*)- (g z^*)^2 = 0, \quad x\in \Omega, \\
        \displaystyle
        \pdv{f_2(x)}{n} = 0, \quad x\in \partial\Omega,
    \end{cases}
    \end{equation}
    and $f_3$ is a continuous on $\Omega$ function which is bounded in $\bar{\Omega}$.
\end{Lemma}

\begin{proof}
    Let 
    \begin{equation}
        F (d, z_1^*) = \nabla \cdot(a_2 \nabla z_1^*) + (Q+d g(x))z_1^*(m - (Q+d g(x))z_1^*),\quad x\in \Omega,
    \end{equation}
    where by \eqref{z_1} it is clear that $F(d,z_1^*) = 0$. We also know that $(0,z^*)$ is a point on the curve $F$. To be able to expand $z_1^*$ as a function of $d$, i.e., $z_1^* = z_1^*(d)$, we will apply the implicit function theorem. It is trivial to see that $F$ is continuously differentiable with respect to both $d, z_1^*$. Further,
    \begin{equation*}
        \displaystyle \pdv{F}{z_1^*} = \nabla\cdot(a_2 \nabla 1) + m (Q + d g) - 2 z_1^* (Q + d g)^2,
    \end{equation*}
    and hence
    \begin{equation*}
        \displaystyle \pdv{F}{z_1^*} \bigg|_{z_1^* = z^*, d = 0} = Q(m - 2 z^* Q) < 0,
    \end{equation*}
    where the last inequality is true since $z^* > m/2Q$ by Lemma \ref{z_star_bound}. Thus the Implicit Function Theorem applies, and we can state that $\exists h>0$ such that if $|d|<h$ then there exists a $z_1^* = z_1^*(d)$ such that $F(d, z_1^*(d)) = 0$. Further we can describe $z_1^*$ by expressing it in the following expansion form $z_1^*(d) = z^* + f_1(x) d +  f_2(x) d +  f_3(x) d^3$. To find $f_1, f_2, f_3$ we can substitute   $z_1^*(d) = z^* + f_1(x) d +  f_2(x) d +  f_3(x) d^3$ into $F(d, z_1^*(d)) = 0$ and obtain
    \begin{equation*}
        \begin{aligned}
            \nabla \cdot (a_2 \nabla (z^* + f_1(x) d +  f_2(x) d^2 +  f_3(x) d^3)& \\
            + (Q + d g)(z^* +  f_1(x) d + f_2(x) d^2 &\\ +  f_3(x) d^3) (m - (Q + d g)(z^* +  f_1(x) d +  f_2(x) d^2 +  f_3(x) d^3))& = 0.
        \end{aligned} 
    \end{equation*}
    We can then arrange the terms according to the powers of $d$,
    \begin{equation*}
        \begin{aligned}
            \nabla \cdot (a_2 \nabla z^*) + Q z^* (m - Q z^*) = 0\\
            d\bigg[ (a_2 \nabla f_1(x)) + m Q f_1(x) - 2 f_1(x) Q^2 z^* + g m z^* - 2 g Q (z^*)^2\bigg] = 0\\
            d^2 \bigg[ \nabla \cdot (a_2 \nabla f_2(x)) + f_1(x) (g m - f_1(x) Q^2 - 4gQz^*) + f_2(x) Q (m - 2Q z^*)- (g z^*)^2 = 0 \bigg] = 0
        \end{aligned}
    \end{equation*}
    and $f_3(\cdot ;d)$ is some function which is uniformly bounded for small $d$. The result follows immediately.
\end{proof}

\begin{Lemma}
Suppose that $Q_1$ is as in   \eqref{Q1_form}, \eqref{normal} and \eqref{normal_add} hold, $m_1\equiv m_2 \equiv m$, $\displaystyle \nabla \cdot \left[ a_2(x)  \nabla \left( \frac{\displaystyle m(x)}{\displaystyle Q(x)}\right) \right] \not\equiv 0$ on $\Omega$, 
and $m(x)\equiv \alpha P+\beta Q$ for some $\alpha,\beta>0$. Additionally assume that
\begin{equation}
\label{cond:lemma}
    \displaystyle \frac{\sup_{x\in\Omega} \frac{m}{Q}}{\inf_{x\in\Omega} \frac{m}{Q}} <2.
\end{equation}
Then $\forall \varepsilon > 0$, $\exists \delta >0$ such that if $|d|<\delta$, then $|v^* - v_1^*| < \varepsilon$.
\end{Lemma}

\begin{proof}
By Lemma \ref{expand_z}, for small $d$, $z_1^* = v_1^*/Q_1$ can be expanded as $z_1^* = z^*+ f_1(x) d +  f_2(x) d^2+  f_3(x)d^3$. Because $z_1^*,z^*$ are continuous functions, and $f_1,f_2$ are also continuous, then $f_3$ must also be continuous. Thus, all the functions are bounded on the closed domain $\bar{\Omega}$ for small $d$.

We have
\begin{equation*}
\begin{aligned}
|v^* - v_1^*| &= |Q z^* - Q_1 z_1^*|\\
&= |Q z^* - (Q + d g)(z^* + d f_1 + d^2 f_2 + d^3 f_3)|\\
&= |(-1)(d (g z^* + f_1 Q) + d^2 (f_1 g + f_2 Q) + d^3 (f_2 g + f_3 Q) + d^4 (f_3 g))|\\
&\leq \bigg| d \max_{x\in\Omega} (g z^* + f_1 Q)\bigg| + \bigg| d^2 \max_{x\in\Omega} (f_1 g + f_2 Q)\bigg| + \bigg| d^3 \max_{x\in\Omega} (f_2 g + f_3 Q)\bigg| + \bigg| d^4 \max_{x\in\Omega} (f_3 g) \bigg|\\
&\leq \bigg| d \max_{x\in\Omega} (g z^* + f_1 Q)\bigg| + \bigg| d \max_{x\in\Omega} (f_1 g + f_2 Q)\bigg| + \bigg| d \max_{x\in\Omega} (f_2 g + f_3 Q)\bigg| + \bigg| d \max_{x\in\Omega} (f_3 g) \bigg|\\
&= |d| \bigg[ \max_{x\in\Omega} |g z^* + f_1 Q| + \max_{x\in\Omega} |f_1 g + f_2 Q| + \max_{x\in\Omega} |f_2 g + f_3 Q| +  \max_{x\in\Omega} |f_3 g| \bigg]
< \varepsilon
\end{aligned}
\end{equation*}
when
\begin{equation*}
\displaystyle \delta = \displaystyle \frac{\varepsilon}{\displaystyle \bigg[ \max_{x\in\Omega} |g z^* + f_1 Q| + \max_{x\in\Omega} |f_1 g + f_2 Q| + \max_{x\in\Omega} |f_2 g + f_3 Q| +  \max_{x\in\Omega} |f_3 g| \bigg]}.
\end{equation*}
\end{proof}

\subsection{Proof of Theorem \ref{coexist_dev}}

\begin{proof}{(Theorem \ref{coexist_dev})}
Recall that \eqref{system2} is a monotone dynamical system. By Lemma \ref{zero_equilibrium}, $(0,0)$ is an unstable repelling equilibrium of \eqref{system2}. For $Q_1 = Q + d g(x)$,
\begin{equation*}
\int_{\Omega} |Q - Q_1| dx = \int_{\Omega} |d g(x)| dx = |d| \int_{\Omega} |g(x)|dx. 
\end{equation*}
Now by Lemma \ref{semi_perturbation}, if $|d| \int_{\Omega} |g(x)| dx < \varepsilon_d$ then $(u^*, 0)$ is unstable. Thus if 
\begin{equation*}
|d| < \frac{\varepsilon_d}{\int_{\Omega} |g(x)| dx}
\end{equation*}
then $(u^*,0)$ is unstable.

By Lemma \ref{lemma_new_deviation} there exists an $\varepsilon$, say $\varepsilon_v$ such that if $|v^* - v_1^*| < \varepsilon_v$ then the equilibrium $(0,v_1^*)$  (where $v_1^*$ is the solution of \eqref{eq_v} with $Q_1$ substituted for $Q$) is unstable.

Choose $d$ such that
\begin{equation*}
|d| < \displaystyle \min \displaystyle \left[ \frac{\varepsilon_d}{\int_{\Omega}|g(x)|dx}, \frac{\varepsilon_v}{\displaystyle \bigg[ \max_{x\in\Omega} |g z^* + f_1 Q| + \max_{x\in\Omega} |f_1 g + f_2 Q| + \max_{x\in\Omega} |f_2 g + f_3 Q| +  \max_{x\in\Omega} |f_3 g| \bigg]} \right].
\end{equation*}
Since $(0,0)$ is a repeller, $(u^*,0), (0,v_1^*)$ are unstable, a coexistence equilibrium exists and is a global attractor.
\end{proof}

\end{document}